\documentclass[review]{lmcs}
\pdfoutput=1

\usepackage{lastpage}
\lmcsdoi{22}{1}{11}
\lmcsheading{}{\pageref{LastPage}}{}{}%
{Jan.~05,~2024}{Feb.~16,~2026}{}

\usepackage{booktabs} 
\usepackage{ amssymb }
\usepackage{amsmath}
\usepackage{mathrsfs}
\usepackage{float}
\usepackage{proof}
\usepackage{subcaption}
\usepackage{verbatim}
\usepackage[utf8]{inputenc}
\usepackage{ mathrsfs }
\usepackage{algorithm2e}
\usepackage{wrapfig}
\usepackage{csvsimple,booktabs}
\AtBeginDocument{%
  }

\begin{document}
\newcommand{\pist}{\pi^{ST}}
\newcommand{\ita}{ITS}
\newcommand{\St}{\pr^{-1}}
\newcommand{\Ss}{S}
\newcommand{\Loc}{\V_{O \cup H}}
\newcommand{\Ls}{Loc_\mathcal{S}}
\newcommand{\st}{st}
\newcommand{\state}{state}
\newcommand{\stof}{state_{of}}
\newcommand{\fstof}{\overline{state}_{of}}
\newcommand{\R}{\mathcal{R}}
\newcommand{\Rs}{\mathcal{R}^{sgn}}
\newcommand{\Rm}{\mathcal{R}^-}
\newcommand{\Rp}{\mathcal{R}^+}
\newcommand{\Ropt}{{\mathcal{R}^{\theta}}}
\newcommand{\Roptm}{{\Ropt}^-}
\newcommand{\Roptp}{{\Ropt}^+}
\newcommand{\Ropts}{{\Ropt}^{sgn}}
\newcommand{\Told}{T}
\newcommand{\Td}{\mathcal{R}^\mathcal{F}}
\newcommand{\M}{\mathcal{M}}
\newcommand{\Pp}{\M}
\newcommand{\I}{\mathcal{I}}
\newcommand{\T}{\mathcal{T}}
\newcommand{\W}{\mathcal{W}}
\newcommand{\A}{\mathcal{A}}
\newcommand{\V}{V}
\newcommand{\Vi}{\V^I}
\newcommand{\Vo}{\V^O}
\newcommand{\Vh}{\V^H}
\newcommand{\minf}{\models}
\newcommand{\mt}{\models_{t}}
\newcommand{\mts}{\models_{t^{sgn}}}
\newcommand{\mtn}{\models_{t^-}}
\newcommand{\mtp}{\models_{t^+}}

\newcommand{\notmt}{\nvDash_{t}}
\newcommand{\notmtn}{\nvDash_{t^-}}
\newcommand{\notmtp}{\nvDash_{t^+}}

\newcommand{\pred}{Pred}
\newcommand{\fun}{func}
\newcommand{\term}{u}
\newcommand{\E}{E}
\newcommand{\Ei}{\E_I}
\newcommand{\Eo}{\E_O}
\newcommand{\Eh}{\E_H}
\newcommand{\unt}{U}
\newcommand{\rel}{R}
\newcommand{\z}{Z}
\newcommand{\y}{Y}
\newcommand{\since}{S}
\newcommand{\always}{G}
\newcommand{\future}{F}
\newcommand{\hist}{H}
\newcommand{\nxt}{X}
\newcommand{\tr}{Tr}
\newcommand{\trp}{\tr^+}
\newcommand{\trm}{\tr^-}
\newcommand{\weak}{\mathcal{W}2\mathcal{S}}
\newcommand{\trs}{\tr^{sgn}}
\newcommand{\run}{run}
\newcommand{\e}{end}
\newcommand{\F}{\mathcal{F}}
\newcommand{\SF}{\mathcal{SF}}
\newcommand{\La}{\mathcal{L}}
\newcommand{\Laf}{\mathcal{L}^{<\omega}}
\newcommand{\Lainf}{\mathcal{L}^{\omega}}
\newcommand{\atn}{@\tilde{F}}
\newcommand{\atnw}{@F}
\newcommand{\atl}{@\tilde{P}}
\newcommand{\fcond}{\psi_{cond}}
\newcommand{\Fcond}{\Psi_{cond}}
\newcommand{\Comp}{\Pp}
\newcommand{\Prop}{\varphi}
\newcommand{\struct}{M}

\newcommand{\pr}{Pr}
\newcommand{\map}{map}
\newcommand{\trr}{TrR}
\newcommand{\trrufa}{TrRuFA}
\newcommand{\trrf}{TrR+F}

\newcommand{\size}[1]{|#1|}
\newcommand{\ite}{ite}

\newcommand{\KWD}[1]{\ensuremath{\textit{#1}}\xspace}
\newcommand{\Tail}{\KWD{Tail}}
\theoremstyle{plain}\newtheorem{ifr}[thm]{Inference}
\theoremstyle{plain}\newtheorem{example}[thm]{Example}

\usetikzlibrary{shapes,arrows}
\usetikzlibrary{arrows.meta}

\tikzstyle{block} = [rectangle, draw, fill=blue!20, 
    text width=3em, text centered, rounded corners, minimum height=1em]
\tikzstyle{line} = [draw, -latex']
\tikzstyle{cloud} = [draw, ellipse,fill=red!20, node distance=1cm,
    minimum height=1em]

\usetikzlibrary{calc, 3d, mindmap, trees, arrows,decorations.pathreplacing,backgrounds,positioning,fit,petri, patterns, shapes, quotes}

\title[Async Composition of LTL Properties over
INF and Finite Traces]{Asynchronous Composition of LTL Properties over
Infinite and Finite Traces} 

\author[A.~Bombardelli]{Alberto Bombardelli \lmcsorcid{0000-0003-3385-3205}
}[a,b]
\author[S.~Tonetta]{Stefano Tonetta \lmcsorcid{0000-0001-9091-7899}}[a]

\address{Fondazione Bruno Kessler, via Sommarive 18, Trento Italy 380123}
\email{abombardelli@fbk.eu, tonettas@fbk.eu}

\address{University of Trento, via Sommarive 9, Trento Italy 380123}
\email{alberto.bombardell-1@unitn.it}

\begin{abstract}
The verification of asynchronous software components poses significant
challenges due to the way components interleave and exchange
input/output data concurrently. Compositional strategies aim to
address this by separating the task of verifying individual components
on local properties from the task of combining them to achieve global
properties. This paper concentrates on employing symbolic model
checking techniques to verify properties specified in Linear-time
Temporal Logic (LTL) on asynchronous software components that interact
through data ports. Unlike event-based composition, local properties
can now impose constraints on input from other components,
increasing the complexity of their composition. We consider both the
standard semantics over infinite traces as well as the truncated
semantics over finite traces to allow scheduling components only
finitely many times.

We propose a novel LTL rewriting approach, which converts a local
property into a global one while considering the interleaving of
infinite or finite execution traces of components. We prove the
semantic equivalence of local properties and their rewritten version
projected on the local symbols. The rewriting is
also optimized to reduce formula size and to leave it unchanged when
the temporal property is stutter invariant. These methods have been
integrated into the OCRA tool,  as part of the contract
refinement verification suite. Finally, the different composition
approaches were compared through an experimental evaluation that
covers various types of specifications.
\end{abstract}

\maketitle

\section{Introduction}

Model checking asynchronous software poses significant challenges 
%
%
due to the non-determini\-stic
interleaving of components and concurrent access to shared
variables.
Compositional techniques are often used to tackle scalability issues.
The idea of this approach is to decouple the problem of verifying local properties
specified over the component interfaces from the problem of composing them to ensure
some global property.

For example, in \cite{RBH+01} is described the following compositional reasoning.
Given some local component ($\M_1, \dots, \M_n$), 
some local properties of these components ($\varphi_1, \dots, \varphi_n$),
a global property ($\varphi$), a notion of composition for the components ($\gamma_S$)
and a notion of composition for the properties ($\gamma_P$); if the local components satisfy
the local properties and the composition of local properties entails the global
properties, the local component composition satisfies the global property. 
While in the synchronous setting, the composition is simply giving by
a conjunction, the asynchronous composition of local
temporal properties may be tricky when considering software components
communicating through data ports.

In this paper, we define the asynchronous composition of LTL~\cite{ltl} properties local to
components; which means that the local properties reason only over the part of the
execution of the local component, e.g., when a local property refers to the
\textit{next} state, the composition will consider the next state
along the local run of the component.
The communication between local components is achieved
with I/O data ports while their execution is controlled by scheduling constraints. 
Due to uncontrollable input changes,
the composition must take into account whether or not the local component is running
even if the formula does not contain ``next''\footnote{Usually
    when an LTL formula does not contain next, it is \textit{stutter invariant}. 
    In such cases, the asynchrony would not interfere with the formula.}.
Another important factor to consider is the
possibly finite execution of local components. If for instance, the scheduler is
not fair, the unlucky local component might be scheduled only for a finite amount
of time. This \textit{``possibly finite''} local semantics also allows for a natural
way to represent permanent faults in asynchronous systems (e.g. a local component
is not scheduled anymore means that it crashed).

To represent the \textit{``possibly finite''} local semantics,
we use the truncated weak semantics defined in \cite{Eisner2003ReasoningWT}
to represent local properties.
The idea is that the local formula can be interpreted to both finite and infinite trace, and the formula semantics does not force the system to execute
the local component. For completeness, we define the composition
in a way that soundly supports the finite execution of the composed system as
well; therefore, the compositional reasoning applies to hierarchical systems
too.

Our composition approach is based on a syntactical rewriting $\R^*_c$ of the local temporal property
to a property of the composite model. The rewriting is then conjoined with a constraint
ensuring that output variables do not change when the component is not running ($\fcond$).
The property composition $\gamma_P$
is in the following form: $\fcond \wedge \bigwedge_{1 \le i \le n} \R^*_i(\varphi_i)$.
We also provide an optimized version of the rewriting that works when
all components run infinitely often, thus without possible truncation
of the local traces.

The proposed approach has been implemented inside OCRA\cite{ocra}, which allows
a rich extension of LTL and uses a state-of-the-art model checking
algorithm implemented in nuXmv~\cite{nuxmv} as back-ends to check
satisfiability.

We evaluated our approach on models representing compositions
of pattern formulas and real models from the automotive domains~\cite{CimattiTACAS23}.
We show both the qualitative differences in results between the two semantics
and the impact of the optimization when dealing with local infinite executions.

\subsection{Motivating examples}
We propose two examples to motivate our work. The first
model is a toy example representing a system that tries to send a value to a network
while the second is a real model coming from the automotive domain. Both models
can be naturally represented through possibly finite scheduling of
local components.

\subsubsection{Sender}
\label{sec:simple_ex}
We now model a three-component system that represents a system that receives
a message and tries to send it through a network. The component either successfully delivers the message or fails to send it and logs the message.
The example is represented in Figure \ref{fig:ex}.
This example is composed of three components:
Component $c1$ receives a message $rec_1$ and an input $in_1$ and tries
to send the value to component $c2$. If $c1$ is eventually able to send $send_2$
message through the network, then $c2$ will run and will output the original input.
The network guarantees that eventually, $c1$ will be able to send the message to $c2$;
however, if $c1$ runs only finitely many times, $c3$ at some point will report a failure.
The global property states that if an input message is received, it is either eventually
delivered as output or an error occurs.
If $c1$ runs infinitely often, the global property is satisfied without
the need for $c3$.

\begin{figure}
  \begin{subfigure}{\textwidth}
    \centering
      \begin{tikzpicture}[>=latex,font=\sffamily, node distance=5em, minimum width=4.5em,
        minimum height=3em,scale=.7]
        \node[rectangle, draw] (c1) at (0, 0) {$c1$};
        \node[rectangle, draw, right = 3.5em of c1] (c2) {$c2$};
        \node[below of =c2] (tmp) {};
        \node[rectangle, draw, left= -2em of tmp] (c3) {$c3$};

        \node[left of=c1] (rec1) {\color{red}$rec_1$};
        \node[below = -15pt of rec1] (in1) {\color{red}$in_1$};
        \node[right of=c1, yshift=10pt, xshift=-18pt] (out1) {\color{blue}$out_1$};
        \node[left of=c2, yshift=10pt, xshift=+22pt] (in2) {\color{red}$in_2$};
        \node[right of=c3, yshift=10pt] (outf) {\color{blue}$out_f$};
        \node[below right of=c3, yshift=20pt] (f) {\color{blue}$fail$};
        \node[below right of=c2, yshift=20pt] (f) {\color{blue}$send$};
        \node[below of=c1, yshift=25pt, xshift=10pt] (try1) {\color{blue}$try_1$};
        \node[left of=c3, yshift=10pt, xshift=15pt] (try) {\color{red}$try$};
        \node[above right of=c1, yshift=5pt, xshift=-10pt] (send1) {\color{blue}$send_1$};
        \node[above left of=c2, yshift=5pt, xshift=10pt] (rec2) {\color{red}$rec_2$};
        \node[right of=c2, yshift=5pt, xshift=-10pt] (o) {\color{blue}$out$};
        \node[above of=c3, yshift=-20pt, xshift=-10pt] (i3) {\color{red}$in_3$};
        \draw[->] (c1) -- (c2);
        \draw[->] (c1.north) -- ++(0, 1) -| (c2);
        \draw[->] (c1) -| (c3.north);
        \draw[->] (c1) |- (c3);
        \draw[<-] (c1) -- ++(-2, 0);
        \draw[<-] (c1.south west) -- ++(-1.5, -1);
        \draw[->] (c2) -- ++(2, 0);
        \draw[->] (c2.south) -- ++(0, -0.5) -- ++(2, 0);
        \draw[->] (c3) -- ++(3, 0);
        \draw[->] (c3.south) -- ++(0, -0.5) -- ++(3, 0);
      \end{tikzpicture}
  \end{subfigure}
  \begin{subfigure}{\textwidth}
    \begin{align*}
      \varphi_{c1} :=&\always (rec_1 \rightarrow out_1' = in_1 \wedge
        \nxt ((try_1\wedge out_1'=out_1) \unt send_1))\\
      \varphi_{c2} :=& \always (rec_2 \rightarrow  out2' = in2 \wedge \nxt send_2)\\
      \varphi_{c3} :=& \always (try \rightarrow out_f' =in_3 \wedge \nxt fail)\\
      \varphi :=&\always ((rec_1 \wedge in_1=v) \rightarrow\future (send \wedge out=v \vee fail \wedge out_f=v))\\
      \alpha :=&\always (in_1 \rightarrow run_1) \wedge\always (send_1 \rightarrow run_2) \wedge
                 \always(\hist^{\le p} try_1 \rightarrow run_3)
    \end{align*}
  \end{subfigure}
  \caption{Figure representing the sender model. The colour red represents input variables while
  the colour blue represents the output variables.}
  \label{fig:ex}
\end{figure}

\subsubsection{Automotive compositional contract}
\label{sec:autosar_ex}
Another interesting example comes from the automotive domain. In \cite{CimattiTACAS23}, 
the EVA framework was proposed for the compositional verification
of AUTOSAR components. That work used the rewriting technique we proposed in \cite{nfm22}
to verify the correct refinement of contracts defined as a pair of LTL properties.

Due to the complexity of the model, we omit a full description of the system,
the specifications and the properties. We focus on a specific requirement of the
system that states that the system shall brake when the Autonomous Emergency Braking
module gets activated. A simplified version of the specification
is defined by the following LTL formula:
$$\always (\nxt (aeb\_breaking.status \neq 0) \rightarrow \future^{\le 2} Brake\_\_In\_\_BrakeActuator \neq 0)$$
The specification is entailed by a brake actuator component that is composed of
the actual actuator (BrakeActuator\#BA\_Actuator) and a watchdog (BrakeActuator\#BA\_Watchdog).
The actuator is scheduled every time a signal is received in input while the watchdog
is scheduled periodically. By reasoning over finite executions of components,
it is possible to verify whether or not the global specification is valid even if at some
point the Actuator stop working.
We omit the detailed structure and specification of the sub-components (actuator and
watchdog), for a complete view please refer to~\cite{CimattiTACAS23} or to the
experimental evaluation.
\subsection{Overall contribution}
The main contribution of this paper is the definition of a rewriting-based technique
to verify compositional asynchronous systems in a general way. Furthermore, we
provide an additional optimized rewriting technique to cover the case in which
local components are assumed to run infinitely often.
The main advantages of our compositional approach are the following:
\begin{itemize}
  \item It supports asynchronous communication between data ports.
  \item It supports generic scheduling constraints expressible through LTL formulas.
  \item It supports both finite and infinite executions of local components making
    it a suitable approach for safety assessment as well.
\end{itemize}
\noindent 
This work is an extension of the conference paper \cite{nfm22}. The work has been
extended with the following contributions:
\begin{itemize}
  \item A weak semantics for the logic with two decision procedures for its verification.
  \item A new definition of asynchronous composition of Interface Symbolic Transition Systems
    that deals with possibly finite execution of components.
  \item A new rewriting that generalizes the previous one for possibly finite systems;
    the rewriting introduced in~\cite{nfm22} is then presented as an optimized 
    version of our new rewriting when each component is scheduled infinitely often.
  \item A new experimental evaluation with new models and a set of benchmarks from
    the automotive domain.
\end{itemize}
\subsection{Outline}
The rest of the paper is organized as follows: in
Sec.~\ref{sec:rw}, we compare the proposed solution with related works;
in Sec.~\ref{sec:logic}, we define our logic syntax, semantics and verification;
in Sec.~\ref{sec:problem}, we formalize the problem; in
Sec.~\ref{sec:rewr}, we define the rewriting approach, its basic version, its
complete version and an optimized variation that is suited for infinite executions
only; in Sec.~\ref{sec:perfeval}, we report on the
experimental evaluation; finally, in
Sec.~\ref{sec:conclusions}, we draw the conclusions and some
directions for future works.

\section{Related works}
\label{sec:rw}

One of the most important works on temporal logic for asynchronous systems is
Temporal Logic of Action (TLA)~\cite{tla} by Leslie Lamport, later extended 
with additional operators\cite{tla+}. TLA has been also used in a component-based
manner in \cite{comptla+}. Our formalism has similarities with TLA. We use
a (quantifier-free) first-order version of LTL~\cite{MannaPnueli92} with ``next''
function to specify the succession of actions of a program. TLA natively supports
the notion of stuttering for composing asynchronous programs so that the
composition is simply obtained by conjoining the specifications.
We focus instead on local properties that are specified independently of how
the program is composed so that ``next'' and input/output data refer only to the local
execution.
Another substantial difference with TLA is that our formalism also supports a finite semantics of LTL,
permitting reasoning over finite traces.

As for propositional LTL, the composition of specifications is studied
in various papers on assume-guarantee reasoning (see, e.g.,
\cite{DBLP:conf/charme/McMillan99a,DBLP:conf/spin/PasareanuDH99,AGltl,cref})
for both synchronous and asynchronous composition.
In the case of asynchronous systems, most works focus on fragments of
LTL without the next operator, where formulas are always stutter
invariant. Other studies investigated how to tackle
down state-space explosion for that scenario usually employing techniques such as
partial order reduction \cite{oldrewr}. However, our work covers a
more general setting, where the presence of input variables makes
formulas non-stutter-invariant.

Similar to our work,
\cite{oldrewr} considers a rewriting for LTL with events to map local
properties into global ones with stuttering. However, contrary to this
paper, it does not consider input variables (nor first-order
extension) and assume that every variable does not change during
stuttering, resulting in a simpler rewriting. In \cite{EisnerFHMC03},
a temporal clock operator is introduced to express properties related
to multiple clocks and, in principle, can be used to interpret
formulas over the time points in which a component is not
stuttering. Its rewriting is indeed similar to the basic version
defined in this paper, but is limited to propositional LTL and has not
been conceived for asynchronous composition. The optimization that we
introduce to exploit the stutter invariance of subformulas results in
simpler formulas easy to be analyzed as shown in our experimental
evaluation.

The rewriting of asynchronous LTL is similar to the transformation of
asynchronous symbolic transition systems into synchronous ones
described in \cite{hydi}. That work considers connections based on
events where data are exchanged only upon synchronization (allowing
optimizations as in shallow synchronization~\cite{shallowsync}). Thus,
it does not consider components that read from input variables that
may be changed by other components. Moreover, \cite{hydi} is not able
to transform temporal logic local properties in global ones as in this
paper.

In \cite{CongLiu22}, the authors defined a technique to verify asynchronous assume/guarantee
systems in which assumptions and guarantees are defined on a safety fragment of
LTL with predicates and functions. The main differences between that work and this
one are the following. The logic considered by \cite{CongLiu22} is limited to a
safety fragment of LTL with ``globally'' and past operators; on the other hand, our work covers
full LTL with past operators with event freezing functions and finite semantics.
Another key difference is that they consider scheduling in which a component is
first dispatched and then, after some steps it completes its action; instead, our framework defines scheduling constraints with  LTL formulas
and actions are instantaneous.

Another related work is the one proposed in~\cite{MBATVA14}. Here, the authors
propose a compositional reasoning based on assume-guarantee contracts for model based safety
assessment (MBSA). They extend the model with possible faults that
are modelled as input Boolean parameter to the system. However,
contrary to our work, asynchronous composition is not considered and a
fault represents the non-satisfaction of a local property. 

In~\cite{BaumeisterCBFS21}, a related rewriting is proposed in the context of
Asynchronous Hyperproperties. That work considers an hyper-logic based on LTL
that uses a rewriting to align different traces of the same systems on ``interesting'' points.
Although for certain aspects the work is similar, it considers a very different problem.

Finally, our logic semantics is based on the works of Eisner, Fisman et al. on truncated
LTL of \cite{Eisner2003ReasoningWT}. The main differences are the following. Our logic includes past operators,
first-order predicates and functions; we consider only weak/strong semantics for
our logic while their work instead also deals with neutral semantics.

In summary, while existing works address various aspects related to our approach,
none provide a unified framework for the composition of asynchronous systems with
I/O data ports that explicitly accounts for the termination of local components.
To the best of our knowledge, this work is the first to bridge this gap, offering
a comprehensive solution that integrates these considerations into a general framework.

\section{First Order Past LTL with weak truncated semantics}
\label{sec:logic}
This section presents the syntax and the semantics of the logic used in this paper and its
verification.

The logic is an extension of LTL with event-freezing functions of \cite{Tonetta17}
that is interpreted over both finite and infinite traces. As we mentioned in the 
introduction, we are considering weak semantics for our logic that is based
on the works of Eisner and Fisman~\cite{Eisner2003ReasoningWT,FismanIsoLA18}.
The overall idea behind this logic is to have an expressive language that can reason on both infinite and truncated
executions of programs.
Finally, we can observe that for each infinite trace,
if the trace satisfies a property, then all its finite prefixes satisfy it as
well under the weak semantics.

Prior to the logic, we denote the notion of traces, which represent
finite or infinite executions of input/output
components. Therefore, traces distinguish between input
and output symbols. A local component reads the input to decide the
next state and output. Thus, the finite traces in our setting
do not contain the evaluation of input variables at the end of the trace.

\begin{defi}
  We define a trace $\pi$ as a sequence $s_0, s_1, \dots$ of assignments over a set of 
  input variables $\Vi$ and output variables $\Vo$ i.e. in which each $s_i$ is defined
  over $\Vi \cup \Vo$.
  A trace $\pi$ is denoted as finite if the sequence of assignments $s_0, s_1, \dots$
  is finite. Moreover, if a trace is finite, its last assignment is defined
  only over the symbols of $\Vo$.
  We denote the set of traces finite and infinite over  $\Vi, \Vo$ as
  $\Pi(\Vi, \Vo)$.
\end{defi}

\subsection{Syntax}
In this paper, we consider LTL~\cite{MP92} extended with past
operators~\cite{LichtensteinPZ85} (with $\since$ as ``since'' and $\y$
as ``yesterday'') as well as ``if-then-else'' ($ite$)
and ``at next'' ($\atn$), and ``at last'' ($\atl$) operators from
\cite{Tonetta17}. For simplicity, we refer to it simply as LTL.

We work in the setting of Satisfiability Modulo Theory
(SMT)~\cite{Barrett:2009hv} and LTL Modulo Theory (see, e.g., 
\cite{ltlevfr}).
First-order formulas are built as usual by proposition logic
connectives, a given set of variables $V$ and a first-order signature
$\Sigma$, and are interpreted according to a given $\Sigma$-theory
$\mathcal{T}$. We assume to be given the definition of
$\struct,\mu\models_{\mathcal{T}}\varphi$ where $\struct$ is a $\Sigma$-structure, $\mu$
is a value assignment to the variables in $V$, and $\varphi$ is a
formula. Whenever $\mathcal{T}$ and $\struct$ are clear from contexts we
omit them and simply write $\mu\models\varphi$.
\begin{defi}
  Given a signature $\Sigma$ and a set of variables $\V$, LTL formulas
  $\varphi$ are defined by the following syntax:
  $$\varphi := \top|\bot|pred(u_1,\dots,u_n)|\neg \varphi_1| \varphi_1 \vee \varphi_2|X \varphi_1|\varphi_1 U \varphi_2|Y \varphi_1| 
  \varphi_1 S \varphi_2$$
  $$ u := c|x|func(u_1,\dots, u_n)|  next(u_1) | ite(\varphi,u_1,u_2) | u_1 \atn \varphi| u_1 \atl \varphi$$
  where $c$, $func$, and $pred$ are respectively a constant, a
  function, and a predicate of the signature $\Sigma$ and $x$ is a variable in
  $\V$.
\end{defi}

Apart from $\atn$ and $\atl$, the operators are standard. $u \atn
\varphi$ represents the value of $u$ at the next point in time (excluding the current point)
in which $\varphi$ holds. Similarly, $u \atl \varphi$ represents the
value of $u$ at the last point in time in which $\varphi$ holds (excluding the current point).
Figure~\ref{fig:atn} provides an intuitive graphical view of the operator's semantics.

\begin{figure}
\begin{tikzpicture}[node distance=2cm, thick, ->]

\tikzstyle{el} = [circle, draw=black, fill=blue!20, align=center, minimum width=.2cm, minimum height=.2cm]
\tikzstyle{arrow} = [thick,->,>=stealth]

\node[el] (step1) {};
\node[left of=step1] (step0) {\dots};
\node[el, right of=step1] (step2) {};
\node[el, right of=step2] (step3) {};
\node[el, right of=step3] (step4) {};
\node[right of=step4] (stepe) {\dots};

  \node[above=7pt of step1] (u1) {\small $\varphi,u=3$};
  \node[right of=u1] (u2) {$\neg\varphi,u=3$};
  \node[right of=u2] (u3) {$\neg\varphi,u=3$};
  \node[right of=u3] (u4) {$\varphi,u=6$};

  \node[below=7pt of step1] (atn) {\small $\term \atn \varphi = 6$};
  \node[below=7pt of step4] (atn) {\small $\term \atl \varphi = 3$};
\draw[arrow] (step0) -- (step1);
\draw[arrow] (step1) -- (step2);
\draw[arrow] (step2) -- (step3);
\draw[arrow] (step3) -- (step4);
\draw[arrow] (step4) -- (stepe);

\end{tikzpicture}
  \caption{Graphical representation of $\atn$ and $\atl$.}
  \label{fig:atn}
\end{figure}

\paragraph{Notation simplification}
In the following, we assume to have a background theory such that the
symbols in $\Sigma$ are interpreted by an implicit structure $\struct$
(e.g., theory of reals, integers, etc.). We therefore omit $\struct$ to
simplify the notation, writing $\pi, i\mts \varphi$ and $\pi(i)(u)$
instead of respectively $\pi, \struct, i\mts \varphi$ and $\pi^{\struct}(i)(u)$.

Finally, we use the following standard abbreviations:
$\varphi_1\wedge\varphi_2:=\neg(\neg\varphi_1\vee\neg\varphi_2)$,
$\varphi_1 R \varphi_2 := \neg (\neg \varphi_1 U \neg \varphi_2)$ ($\varphi_1$ releases $\varphi_2$),
$F\varphi:= \top U \phi$ (sometime in the future $\varphi$),
$G \varphi:= \neg F \neg\varphi$ (always in the future $\varphi$),
$O \varphi:= \top S \varphi$ (once in the past $\varphi$),
$H \varphi:= \neg O \neg\varphi$ (historically in the past $\varphi$),
$Z \varphi:= \neg Y \neg\varphi$ (yesterday $\varphi$ or at initial state),
$X^n \varphi := X X^{n-1} \varphi$ with $X^0 \varphi := \varphi$,
$Y^n \varphi := Y Y^{n-1} \varphi$ with $Y^0 \varphi := \varphi$,
$Z^n \varphi := Z Z^{n-1} \varphi$ with $Z^0 \varphi := \varphi$,
$F^{\le n} \varphi := \varphi \vee X \varphi \vee \cdots \vee X^n \varphi$,
$G^{\le n} \varphi := \varphi \wedge X \varphi \wedge \cdots \wedge X^n \varphi$,
$O^{\le n} \varphi := \varphi \vee Y \varphi \vee \cdots \vee Y^n \phi$,
$H^{\le n} \varphi := \varphi \wedge Z \varphi \wedge \cdots \wedge Z^n \varphi$.

\begin{example}
We provide an example of our logic from~\cite{Tonetta17} representing a sensor.
In this example, the sensor has an input real variable $y$, an output real variable $x$ and a Boolean flag $correct$ that represents whether or not
the value reported by the sensor is correct.
We specify that $x$ is always equal to the last correct input value
  with $\always (x = y\atl(correct))$. 
We assume that a failure is permanent with  $\always(\neg correct \rightarrow \always \neg correct)$.
Additionally, we consider also a Boolean variable $read$ representing the event of reading $x$.
  In our example, $reading$ occurs periodically with period $5$ ($read \wedge \always (read \rightarrow \nxt (\always^{\le 3} \neg read) \wedge \nxt^5 read)$):
  initially $read$ is true;
  at each step $i$, if $read$ is true, then in the states ranging from $i+1$
  to $i+4$ ($X G^{\le 3}$) $read$ is false; and, finally
  read is true at position $i+5$.
  
Finally, let us say that an alarm $a$ is true if and only if the last two read values are the same: 
  $\always (a\leftrightarrow x\atl (read) = (x\atl(read))\atl(read))$.

We can prove that, given the behaviour defined above, every point at which the sensor is not correct, is followed by an alarm in at most $10$ steps.

\begin{align*}
  &(\always (x=y\atl(correct)) \wedge \always (\neg correct \rightarrow \always \neg correct) \wedge\\
  &read\wedge\always(read \rightarrow \nxt (\always^{\le 3} \neg read) \wedge \nxt^5 read) \wedge\\
  &\always (a \leftrightarrow x \atl(read) = (x \atl(read))\atl (read)))\\
  &\rightarrow \always (\neg correct \rightarrow \future^{\le 10} a)
\end{align*}

\end{example}
\subsection{Semantics}
\label{sec:ltlsem}
We propose an alternative semantics for LTL with past operators and event freezing function 
based on the truncated semantics of LTL defined in \cite{Eisner2003ReasoningWT}.
This semantics is tailored for both finite and infinite traces. For what regards
infinite traces, the semantics is identical to the standard one.
When we reason over finite traces, the predicates are interpreted ``weakly" i.e.
all the predicates are evaluated as true at the end of the trace. 
Moreover, the semantics differentiate between output and input predicates by treating
the latter as \textit{next} operators i.e. they are evaluated as true in the last
state.

Another important difference between standard LTL and other finite semantics
(e.g. LTL$_f$\cite{ltlf}) lies in the negation. When a formula is negated the interpretation
passes from \textit{weak} to \textit{strong} and vice versa. Therefore,
a negated formula $\neg \phi$ is satisfied with weak (strong) semantics if and only if
$\phi$ is not satisfied with strong (weak) semantics. This means for instance
that a proposition $P$ is evaluated as \textit{true} at the end of a trace 
independent of the polarity of its variables. These different semantics are denoted
as $\mtn$ for weak semantics and as $\mtp$ for strong semantics; moreover, we use
$\mts$ to refer to both interpretations.
For what regards the difference with LTL$_f$, LTL$_f$ requires finite traces to
satisfies eventualities i.e. $\pi \models F p$ iff at some point in the trace $p$
is true. However, in our scenario that semantics is not desirable because we would
like to consider as ``good'' traces also traces that are prefixes of traces satisfying
$F p$ (a further discussion on the choice of semantics is given in Section~\ref{sec:sem_just}).

Contrary to the original truncated semantics of \cite{Eisner2003ReasoningWT}, we consider also
past operators, if-then-else, event freezing operators ($\atn, \atl$), functions and terms.

For what regards past operators, the semantics
is similar to standard LTL. The only alteration of the past semantics
is that $\y$ is evaluated as \textit{true} at the end of the trace. This is done
to keep the value of each formula trivially true when evaluated outside of the trace
bounds.

In this paper, terms are not evaluated with a polarity because they are 
interpreted inside predicates. Besides $ite, \atn,\atl$, these operators are 
interpreted identically to standard LTL.

For if-then-else, instead of considering 2 possible term interpretations ($u_1, u_2$),
the new semantics consider a third default value. The third value is chosen
when both the if condition and its negation are satisfied weakly (which is possible
with finite traces).
The idea is that, when both $\varphi$ and $\neg\varphi$ are weakly satisfied, we avoid committing to
either value.
As a result, with this three-valued if-then-else, it is still true that $ite(\varphi, \term_1, \term_2) = ite(\neg\varphi, \term_2, \term_1)$
which would be not true with a two value semantics.

Regarding $\varphi \atn \term$ and $\varphi \atl \term$, the new semantics considers strong satisfaction of the formula to get the value of the term.
This occurs because strong semantics ensures that variables are well-defined in the trace.
When the $\varphi$ is not strongly satisfied, a default value $def_{\varphi \atn \term}$/
$def_{\varphi \atl \term}$ is considered.

From an higher level perspective, the motivation on using the ``weak'' semantics
is briefly described in Section~\ref{sec:sem_just}.

The semantics is defined as follows:
\begin{itemize}
  \item $\pi, \struct, i \mtn pred^O(u_1,\dots, u_n) \text{ iff } |\pi| \le i \text{ or }pred^{\struct}(\pi^{\struct}(i)(u_1),\dots, \pi^{\struct}(i)(u_n))$
  \item $\pi, \struct, i \mtp pred^O(u_1,\dots, u_n) \text{ iff } |\pi| > i \text{ and }pred^{\struct}(\pi^{\struct}(i)(u_1),\dots, \pi^{\struct}(i)(u_n))$
  \item $\pi, \struct, i \mtn pred^I(u_1,\dots, u_n) \text{ iff } |\pi| \le i - 1\text{ or }pred^{\struct}(\pi^{\struct}(i)(u_1),\dots, \pi^{\struct}(i)(u_n))$
  \item $\pi, \struct, i \mtp pred^I(u_1,\dots, u_n) \text{ iff } |\pi| > i + 1 \text{ and }pred^{\struct}(\pi^{\struct}(i)(u_1),\dots, \pi^{\struct}(i)(u_n))$
  \item $\pi, \struct, i\mts \varphi_1\wedge\varphi_2\text{ iff }\pi, \struct, i\mts\varphi_1
\text{ and }\pi, \struct, i\mts\varphi_2$
  \item $\pi, \struct, i\mtn\neg\varphi\text{ iff }\pi, \struct, i\not\mtp\varphi$
  \item $\pi, \struct, i\mtp\neg\varphi\text{ iff }\pi, \struct, i\not\mtn\varphi$
  \item $\pi, \struct, i\mts\varphi_1 {\unt}\varphi_2\text{ iff there exists }
k\geq i, \pi, \struct, k\mts\varphi_2 \text{ and for all }l,
i\leq l<k,\pi, \struct, l$ $\mts\varphi_1$
  \item $\pi, \struct, i\mtn\varphi_1 {\since}\varphi_2\text{ iff } i \ge |\pi| \text{ or there exists }
k \leq i, \pi, \struct, k \mtn\varphi_2\text{ and for all }
  l, k < l \leq i,\pi, \struct, l\mts\varphi_1$
  \item $\pi, \struct, i\mtp\varphi_1 {\since}\varphi_2\text{ iff } i < |\pi| \text{ and there exists }
k \leq i, \pi, \struct, k \mtn\varphi_2\text{ and for all }
  l, k < l \leq i,\pi, \struct, l\mts\varphi_1$
\item $\pi, \struct, i\mts \nxt\varphi\text{ iff } \pi, \struct, i+1 \mts \varphi$
\item $\pi, \struct, i\mtn \y \varphi\text{ iff } i \ge |\pi| \text{ or } i>0 \text{ and } \pi, \struct,
  i-1\mtn \varphi$
\item $\pi, \struct, i\mtp \y \varphi\text{ iff } 0 < i < |\pi| \text{ and } \pi, \struct, i-1\mtp \varphi$
\end{itemize}
where $pred^I$ denotes a predicate containing input variables or $\atn$ terms, $pred^O$ denotes
a predicate that contains only output variables without at next, $sgn \in \{-,+\}$ and $pred^\struct$ is the interpretation in
$\struct$ of the predicate in $\Sigma$.

The interpretation of terms $\pi^{\struct}(i)$ is defined as follows:
\begin{itemize}
\item
  $\pi^{\struct}(i)(c)=c^{\struct}$
\item
  $\pi^{\struct}(i)(x)=s_i(x)$ if $x\in \V$
\item
  $\pi^{\struct}(i)(func(u_1,\dots, u_n))= func^{\struct}(\pi^{\struct}(i)(u_1),\ldots,
    \pi^{\struct}(i)(u_n))$
  \item $\pi^{\struct}(i)(next(u)) = \pi^{\struct}(i+1)(u) \text{ if } |\pi| > i + 1\\
    \pi^{\struct}(i)(next(u)) = def_{next(u)} \text{ otherwise.}$
  \item $\pi^{\struct}(i)(u\atn(\varphi)) = \pi^{\struct}(k)(u) \text{if there exists } k > i
   \text{ such that, for all } l, i < l < k, \pi, \struct, l \mtp \neg \varphi \text{ and } \pi, \struct,  k \mtp \varphi$;\\
    $\pi^{\struct}(i)(u\atn(\varphi)) = def_{u \atn \varphi}$ otherwise.
  \item $\pi^{\struct}(i)(u\atl(\varphi)) = \pi^{\struct}(k)(u) \text{if } i < |\pi| \text{ and there exists } k < i
    \text{ such that, for all } l, i > l > k, \pi, \struct, l \mtp \neg \varphi \text{ and } \pi, \struct,  k \mtp \varphi$;\\
    $\pi^{\struct}(i)(u\atl(\varphi)) = def_{u \atl \varphi}$ otherwise.\\
  \item $\pi^{\struct}(i)(ite(\varphi, u_1, u_2)) = 
    \begin{cases}
      \pi^{\struct}(i)(u_1) & \text{ if } \pi, \struct,  i \mtp \varphi\\
      \pi^{\struct}(i)(u_2) & \text{ if }\pi, \struct,  i \mtp \neg \varphi\\

      def_{ite(\varphi, u_1, u_2)} & \text{ otherwise}
    \end{cases}$
\end{itemize}
where $func^\struct,c^{\struct}$ are the interpretation
$\struct$ of the symbols in $\Sigma$, and $def_{u \atn \varphi}$,
$def_{u \atl \varphi}$ and $def_{ite(\varphi, u_1, u_2)}$ 
are extra variables used to represent default values.

Finally, we have that $ \pi, \struct \mt \varphi \text{ iff } \pi,
\struct, 0 \mtn \varphi$.

%
%
%
\begin{example}
Here we provide a short example to clarify the semantics.
Suppose that a local component is represented
by a temporal property $\varphi_i:= \always (i \rightarrow \nxt o)$; the property expresses that, if the input $i$ is received, at the
next step in the local trace the output $o$ is provided. In our semantics, this means that the execution 
$\pi_f = \{i \}, \{ o\}, \{ o, i\}, \{ o\}$ is satisfied by the property, but since
the semantics is weak also $\pi_f' = \{o,i\} \{o\}, \{i\}$ satisfies the property
according to the semantics. It should be noted
that when we reason over infinite executions, the truncated weak semantics is identical to the standard semantics; therefore, this semantics can be interpreted as a
generalisation of the standard one.
\end{example}

\subsection{Verification}
We present two techniques to verify LTL formulae with truncated semantics.
The first technique is a variation of the rewriting proposed in \cite{ltlf}.
It reduces the verification under finite-trace semantics to standard infinite-trace LTL
semantics by introducing a fresh Boolean variable $Tail$, which is true only at the last position
of the original finite trace and remains true in all subsequent positions.
This effectively marks the end of the trace.
The rewriting consider both traces in which $Tail$ becomes true and traces in
which $Tail$ is never true.
In our setting, the extended trace may assign arbitrary values to the input variables from the last position onward. The key intuition is that, under our semantics, the values of input variables matter only at positions where $\Tail$ evaluates to false i.e. before position $|\pi| -1$.
\begin{defi}
  We define $\tr(\varphi)$ as follows:
  \begin{align*}
    &\trm(\pred^I) := Tail \vee \pred^I\quad \trp(\pred^I) := \neg Tail \wedge \pred^I\\
    &\trs(\pred^O) := \pred^O \quad \quad \trs(\varphi_1 \vee \varphi_2) := \trs(\varphi_1) \vee \trs(\varphi_2)\\
    &\trm(\neg \varphi) = \neg \trp(\varphi) \quad \trp(\neg \varphi) = \neg \trm(\varphi)\\
    &\trm(\nxt \varphi) :=Tail \vee \nxt (\trm(\varphi)) \quad \trp(\nxt \varphi) := \neg Tail
    \wedge \nxt (\trp(\varphi))\\
    &\trs(\varphi_1 \unt \varphi_2) := \trs(\varphi_1) \unt (\trs(\varphi_2))
  \end{align*}
  where $sgn \in \{ -, +\}$ and $\tr(\varphi) := \neg Tail \wedge \always (Tail \rightarrow \nxt Tail \wedge
  \bigwedge_{v_o \in \Vo} v_o = next(v_o)) \wedge \rightarrow
  \trm(\varphi)$ is the top-level rewriting that defines the behaviour of $Tail$

\footnote{
To verify the rewritten formula over a transition system, the system must be extended with a sink state in
which the predicate $Tail$ holds.
This ensures that when $Tail$ becomes true, the system remains in that state indefinitely
(i.e., the transition relation becomes $(Tail \rightarrow {\Vo}' = \Vo \wedge Tail')
\wedge (\neg Tail \rightarrow \T)$).
Note that the rewritten formula implicitly refers to the transition relation of this modified system.
}
\end{defi}
The following theorem relates the truncated semantics to the rewriting. For each finite or infinite trace
$\pi$ we show that the given a formula $\varphi$, $\pi$ satisfies $\varphi$ if and only if a 
\begin{thm}
  Let $\pi$ be a trace over $\Vi,\Vo$ and $\pi'$ be the infinite trace over $\Vi, \Vo \cup \{ Tail\}$
  constructed from $\pi$ by extending the original trace with $Tail$ as follows.
  \begin{itemize}
    \item For all $0 \le i < |\pi| -1: \pi'(i)(\Vi, \Vo) = \pi(i)(\Vi, \Vo), \pi'(i)(Tail) = \bot$;
    \item $\pi(|\pi| - 1)(\Vo) = \pi'(|\pi| -1)(\Vo)$;
    \item for all $j \ge |\pi| -1$: $\pi'(j)(Tail) = \top$, $\pi(j)(\Vo)=\pi(j+1)(\Vo)$.
    \item for all $j \ge |\pi| -1$, values of $\Vi$ are unconstrained i.e. each input variable can be
        assigned to any value. 
  \end{itemize}
  It is true that:
  
  $$\pi \mts \varphi \Leftrightarrow \pi' \models_{LTL} \trs(\varphi)$$
\end{thm}
\begin{proof}
  We can prove the theorem inductively over the structure of the formula for each $i \ge 1$.
  Formally the statement is: for all $i \ge 0$:
  $\pi, i' \mts \varphi \Leftrightarrow \pi, i \models_{LTL} \trp(\varphi)$ where
  $i' = \min(|\pi| -1, i)$.

  The base cases on predicates are very simple.
  If $i' \ge |\pi| -1$ then $\trm(\pred^I)$ is true, $\trp(\pred^I)$ is false and
  $\trs(\pred^O) = \pred^O$ which has the value of $\pi(|\pi|-1)(\pred^O)$ when
  $i \ge |\pi| -1$.

  If $i' < |\pi| -1$ then output variables are immediate to prove and
  input variables follows since $Tail$ is false.

  For the inductive cases, negation, disjunction and Until follow from induction.
  We need to prove only $\nxt$.
  If $i' = |\pi| -1$, then $\pi, i'\mtn \nxt \varphi$ and $\pi, i' \not\mtp \nxt \varphi$;
  since $\pi, i' \models_{LTL} Tail$, then $\pi', i \models_{LTL} \trm(\nxt \varphi)
  $ and $\pi', i \not\models_{LTL} \trp(\nxt \varphi)$.
  If $i' < |\pi| -1$, then the case holds by induction ($Tail$ is false).
\end{proof}

The second technique reduces the problem to the verification of safetyLTL fragment~\cite{Sistla85}.
The idea is that the property is valid on the truncated semantics iff the property
is valid with standard semantics and its "safety part" is valid with truncated semantics.
The verification of the safety fragment is carried out reducing the problem to
invariant checking. Contrary to the standard semantics, in our setting it is not
necessary to check that the finite trace is extendible as done instead in \cite{ifm23}.
The reduction from safetyLTL to invariant is based on \cite{CGH94, LAT03, KV01}
and have been detailed practically in \cite{ifm23,jpk60}. In the following, we propose
the rewriting to transform a formula from LTL to safetyLTL.
The transformation is performed by replacing Until with 
Release on the formula in negative normal form (see Definition \ref{def:weak} below).
The intuition is that with weak semantics and finite traces Until are interpretable
as ``Weak'' Until ($\always \varphi_1 \vee (\varphi_1 \unt \varphi_2)$)
which is easily rewritten as Release ($\varphi_2 \rel (\varphi_1 \vee \varphi_2)$).

\begin{defi}
  \label{def:weak}
Given a formula $\varphi$ in negative normal form, we define the ``Weak-to-Safety''
translation function $\weak(\varphi)$, which maps formulas interpreted under
weak semantics into equivalent formulas interpreted within the SafetyLTL fragment, as follows:
  \begin{align*}
    &\weak(\pred) := \pred
    \quad
    \weak(\neg \pred) := \neg \pred
    \\
    &\weak(\varphi_1 \vee \varphi_2) := \weak(\varphi_1) \vee \weak(\varphi_2)
    \quad \weak(\varphi_1 \wedge \varphi_2) := \weak(\varphi_1) \wedge \weak(\varphi_2)\\
    &\weak(\nxt \varphi) := \nxt \weak(\varphi) \quad \weak(\y \varphi) := \y \weak(\varphi)\\
    &\weak(\varphi_1 \unt \varphi_2) := \weak(\varphi_2) \rel (\weak(\varphi_1) \vee \weak(\varphi_2))\\
    &\weak(\varphi_1 \rel \varphi_2) := \weak(\varphi_1) \rel (\weak(\varphi_2))
  \end{align*}
\end{defi}

\begin{thm}
  For each trace $\pi$:
  \begin{itemize}
    \item If $\pi$ is an infinite trace: $\pi \mtn \varphi \Leftrightarrow \pi \models_{LTL} \varphi$.
    \item If $\pi$ is a finite trace: $\pi \mtn \varphi \Leftrightarrow \pi, \mtn \weak(\varphi)$.
  \end{itemize}
\end{thm}

\begin{proof}
  In the case of infinite traces, the theorem follows from the fact that the truncated
  semantics is equivalent to LTL (with event-freezing functions).

  We prove the finite trace case by induction on the structure of the formula
  considering each possible index $0 \le i < n$ with $n = |\pi|$.

  We only need to prove the case of the operator $\unt$ since $\weak$ is transparent
  for the other operators. Moreover, we need to prove only weak semantics because
  the top level formula is in negative normal form. 

  $\pi, i \mtn \varphi_1 \unt \varphi_2 \Leftrightarrow\exists k \ge i \text{ s.t. }
  \pi, k \mtn \varphi_2 \text{ and }\forall i \le j < k: \pi, j \mtn \varphi_1$.
  By induction we obtain that $\dots \Leftrightarrow \exists k > i \text{ s.t. }
  \pi, k \mtn \weak(\varphi_2) \text{ and } \forall i \le j < k: \pi, j \mtn \weak(\varphi_1)$.

    If $k \ge |\pi|$ then $\pi, k \mtn \varphi_2$ since all formulae are satisfied
    weakly at the end of the trace; therefore, $\pi, i \mtn \always \varphi_1
    \vee \varphi_1 \unt \varphi_2$ which is equivalent to $\varphi_2 \rel (\varphi_1 \vee \varphi_2)$.
\end{proof}

\section{Compositional reasoning}
\label{sec:problem}
\subsection{Formal problem}
\label{sec:mainp}
Compositional verification proves the properties of a system by proving
the local properties on components and by checking that the
composition of the local properties satisfies the global one (see
\cite{RBH+01} for a generic overview). This reasoning is expressed
formally by inference \ref{ifr:comp}, which is parametrized by a
function $\gamma_S$ that combines the component's implementations and
a related function $\gamma_P$ that combines the local properties.

\begin{ifr}
\label{ifr:comp}
  Let $\Comp_1, \Comp_2,\dots,\Comp_n$ be a set of $n$ components, $\Prop_1,\Prop_2,\dots,\Prop_n$ be
  local properties on each component, $\gamma_S$ is a function that defines the
  composition of $\Comp_1,\Comp_2,\dots,\Comp_n$, $\gamma_P$ combines the properties depending
  on the composition of $\gamma_S$ and $\Prop$ a property. The following inference is true:\\
  \[ 
   \infer{\gamma_S(\Comp_1,\Comp_2,\dots,\Comp_n) \models \Prop}{
    \infer{\gamma_S(\Comp_1,\Comp_2,\dots,\Comp_n)\models \gamma_P(\Prop_1,\Prop_2,\dots,\Prop_n)}{
      \Comp_1 \models \Prop_1, \Comp_2 \models \Prop_2,\dots, \Comp_n \models \Prop_n
    }
    & \deduce{\gamma_P(\Prop_1,\Prop_2,\dots,\Prop_n) \models \Prop}{}
}
\]
\end{ifr}
\noindent 
The problem we address in this paper is to define proper $\gamma_S, \gamma_P$
representing the composition of possibly terminating components and temporal properties
such that the inference rule holds.

\subsubsection{Motivation for the Semantic choice}
\label{sec:sem_just}
In this work, we adopt the weak semantics of temporal operators described in Section~\ref{sec:logic}.
This choice is motivated by the characteristic of our setting.
Specifically, during composition, local traces may be truncated by the scheduler
or due to a component failure,
and these truncated traces remain relevant and may contribute to the
satisfaction of the system-level property.
In a sense, we need to pertain the safety part of the property ensured by the
finite trace, while disregarding the liveness part of it.
The weak semantics proposed by \cite{Eisner2003ReasoningWT} allows us
to account for this, ensuring that truncated traces
are not disregarded in the logical reasoning process.

As a motivating example, consider the local property $\varphi_{c1}$ in Figure~\ref{fig:ex}.
This property states that whenever component $c1$ receives an input $in_1$ via the $rec_1$ signal, it will attempt to propagate the message to component $c2$ until it succeeds.
Suppose that $c1$ receives the value $3$ at time $0$ and begins this propagation.
If a failure occurs or the scheduler stops activating $c1$, its execution may be prematurely truncated, and only a finite prefix of its local trace is available.
Still, this partial trace captures the component’s behaviour up to the end of its execution--
specifically, in this case, that $c_1$ initiated the process of sending the input to $c_2$.
Weak semantics ensures that such truncated traces are preserved in system-level reasoning, without requiring that liveness properties be fully realized.
If (on the composition) stronger guarantees are required--e.g., to ensure infinite progress or rule out failures—this can be encoded explicitly in the environment assumption, such as
$\alpha' := \alpha \wedge \always \future \run_{c1}$,
ensuring fair scheduling and infinite execution.

At the global level, traces may also be finite to support hierarchical composition. 
In such cases, weak semantics provides a conservative and practical interpretation: a property satisfied on a finite trace indicates that no counterexample has been found yet, though one may still arise on an extension.
Consider the system-level property
$\varphi := \always ((rec_1 \wedge in_1 = v) \rightarrow \future (send \wedge out = v \vee fail \wedge out_f = v))$
from Figure~\ref{fig:ex}.
In this case, since $\varphi$ is a liveness property, every finite trace satisfy weakly $\varphi$.
In general, the only problematic cases would involve finite traces that actually constitute counterexamples, for instance due to reaching a deadlock state that violates liveness.
However, removing deadlocks is computationally costly, especially in a compositional setting.
Thus, adopting weak semantics is a conscious compromise: it avoids penalizing valid partial behaviours, and in practice, systems are typically defined to avoid such pathological deadlock cases.


%

\subsection{Interface Transition Systems}

In this paper, we represent I/O components as Interface Transition Systems, 
a symbolic version of
interface automata ~\cite{DBLP:conf/sigsoft/AlfaroH01} that considers I/O variables
instead of I/O actions.

\begin{defi}
  \label{def:its}
  An Interface Transition System (ITS) $\M$ is a tuple\\
  $ \Pp = \langle\Vi, \Vo, \I, \T, \SF \rangle$
  where:
  \begin{itemize}
    \item $\Vi$ is the set of input variables while
      $\Vo$ is the set of output variables where $\Vi \cap \Vo= \emptyset$.
    \item $\V := \Vi \cup \Vo$ denotes the set of the variables of $\M$
     \item $\I$ is the initial condition, a formula over $\Vo$, 
     \item $\T$ is the transition condition, a formula over $\V \cup {\Vo}'$
       where ${\Vo}'$ is the primed versions of $\Vo$
    \item $\SF$ is the set of strong fairness constraints, a set of pairs of formulas over $\V$.
  \end{itemize}
  A symbolic transition system with strong fairness
  $\M = \langle \V, \I, \T, \SF \rangle$ is an
  interface transition system without input variables (i.e.,
  $\langle \emptyset, \V,  \I, \T, \SF \rangle$). 
\end{defi}

\begin{defi}
A trace $\pi$ of an ITS $\M$ is a trace
  $\pi = s_0 s_1 \dots \in \Pi(\Vi, \Vo)$ s.t.
  \begin{itemize}
    \item $s_0 \models \I$
    \item For all $i < |\pi| - 1$, $s_{i} \cup s_{i+1}' \models \T$, and
    \item If $\pi$ is infinite: for all $\langle f_A, f_G \rangle \in \SF$, If for all $i$ there exists $j \ge i$, $s_j \models f_A$ then for all $i$ there exists $j \ge i: s_j \models f_G$.
  \end{itemize}
  The finite and infinite languages $\Laf(\M)$ and $\Lainf(\M)$ of an ITS $\M$ is the
  set of all finite and infinite traces of $\M$, respectively.
  Finally, we denote $\La(\M) := \Laf(\M) \cup \Lainf(\M)$ as the language of $\M$.
\end{defi}

\begin{defi}
  Let $\pi= s_0 s_1 \dots$ be a trace of an ITS $\M$ and $\V' \subseteq \V$ a set of symbols
  of $\M$. We denote $s_i(\V')$ as the restriction of the assignment $s_i$ to the symbols
  of $\V'$; moreover, we denote $\pi_{|\V'} := s_0(\V') s_1(\V') \dots$ as the restriction of all the state assignments
  of $\pi$ to the symbols of $\V'$. Furthermore, we denote
  $\La(\M)_{|\V'} = \{ \pi_{|\V'} | \pi \in \La(\M) \}$
  as the restriction of all the traces of the language of an ITS $\M$ to a set of symbols
  $\V' \subseteq \V$.
\end{defi}
\begin{defi}
Let $\M_1, \dots,\M_n$ be $n$ \ita, they are said compatible iff they share respectively only
  input with output (i.e. $\forall i \le n, j \le n \text{ s.t. }i\neq j: \V_i \cap \V_j = (\Vo_i \cap \Vi_j) \cup (\Vi_i \cap \Vo_j)$)
\end{defi}

\subsection{Asynchronous composition of ITS}
We now provide the notion of asynchronous composition ($\otimes$) that will be used
as the composition function ($\gamma_S$) of Inference~\ref{ifr:comp}.
\begin{defi}
\label{def:symbasynccomp}
  Let $\M_1, \dots, \M_n$ be $n$ compatible interface transition systems, let 
  $\run_1,\dots,$ $\run_n$ be $n$ Boolean variables not occurring in $\M_1, \dots, \M_n$
  (i.e. $\run_1,\dots \run_n  \notin \bigcup_{1 \le i \le n} (\Vi_i \cup \Vo_i)$)
  and $\e_1, \dots,\e_n$ be $n$ Boolean variables not occurring in $\M_1, \dots,\M_n$ i.e.
  $\e_1,$ $\dots \e_n  \notin \bigcup_{1 \le i \le n} (\Vi_i \cup \Vo_i)$.
  $\M_1 \otimes \dots \otimes \M_n = \langle \Vi, \Vo, \bigwedge_{1 \le i \le n}{\I_i}, \T, \SF\rangle$ where:
  \begin{align*}
    \Vi =& (\bigcup_{1 \le i \le n}{\Vi_{i} \cup \{ \run_i\}}) \setminus (\bigcup_{1 \le i < n} \bigcup_{i < j \le n} \V_i \cap \V_j)\\
    \Vo =& (\bigcup_{1 \le i \le n}{\Vo_i} \cup \{\e_i\}) \cup (\bigcup_{1 \le i < n} \bigcup_{i < j \le n} \V_i \cap \V_j)\\
    \T =& \bigwedge_{1 \le i \le n}{((\run_i \rightarrow \T_i) \wedge
    \fcond^{\M_i} \wedge (\e_i \leftrightarrow \e_i' \wedge \neg \run_i))}\\
    \SF =&\bigcup_{1 \le i \le n} (\{ \langle \top, \run_i \vee \e_i \rangle \} \cup 
    \{\langle \run_i \wedge \varphi_a,
    \run_i \wedge \varphi_g \rangle | \langle \varphi_a, \varphi_g \rangle \in \SF_{i} \})
  \end{align*}
  where $\fcond^{\M_i} := \neg \run_i \rightarrow  \bigwedge_{v \in \Vo_i}{v=v'}$.
\end{defi}
\noindent 
The operator $\otimes$ provides a notion of asynchronous composition of Interface Transition Systems based
on interleaving. Each component can either run (execute a transition) or stutter
(freeze output variables). Contrary to our previous work\cite{nfm22},
this new definition allows for finite executions of local components. For each $i$,
we introduce
the variable $\run_i$ representing the execution of a transition of  $\M_i$;
furthermore, we introduce the prophecy variables $\e_i$ that monitor whether or not a component will execute new transitions in the future i.e. $\e_i$ is equivalent to $\always \neg \run_i$.

\begin{defi}
  \label{def:pr}
  Let $\pi$ be a trace of $\M=\M_1 \otimes \dots \otimes\M_n$ with $i \le n$.
  We define the projection function $\pr_{\M_i} : \Pi(\V) \rightarrow \Pi(\V_i)$ as follows:
  $$ \pr_{\M_i}(\pi) = \begin{cases}
    s_{\map_0}(\V_i), \dots & \text{ If }\pi \text{ is infinite and } \pi \models \always \future run_i\\
    s_{\map_0}(\V_i), \dots, s_{\map_{n - 1}}(\V_i), s_{\map_n}(\Vo_i) & \text{ Otherwise }
  \end{cases}$$
  where $\map_k$ is the sequence mapping each state of $\pi_i$ into a state
  of $\pi$ as follows:
  $$ \map_k :=
  \begin{cases}
    k & \text{ If } k < 0\\
    \map_{k-1} + 1 & \text{ If } k \ge |\pi| - 1\\
    k' \mid k'> \map_{k - 1}, \pi, k' \models \run_i \text{ and }
    \forall_{\map_{k-1} < j' < k'} \pi, j' \not\models \run_i & \text{ Otherwise}
  \end{cases}
  $$
  Moreover, we define the inverse operator of $\pr$, denoted by $\St$:
  $$\St_{\M_i}(\pi) = \{ \pi' | \pr_{\M_i}(\pi') = \pi\}$$
\end{defi}
Definition \ref{def:pr} provides a mapping between the composed ITS and the local
transition system. The function $\map_k$ maps indexes of the local trace to the
indexes of the global trace. For each $k$ in the range of the trace excluding the
last point (i.e. $[0, |\pi| - 1)$) $\map_k$ represent the $k$-th occurrence of
$\run$ counting from $0$. It should be noted that $\map_0$ is not $0$ but the first 
point of the trace in which $\run_i$ is true ($k' | \map_{-1} < k', \pi, k' \models\run_i$ 
and $\forall \map_{-1} < k'' < k'$ $\pi \not\models \run_i$). Finally, the projection
of a global trace to a local trace is defined using $\map$.
A graphical representation of projection is shown in Figure
\ref{fig:pr}. We now show that the language of each local ITS $\M_i$ contains
the language of the composition projected to the local ITS. From this result
we can derive a condition for $\gamma_P$ such that Inference \ref{ifr:comp}
holds. We do that by considering a mapping between local trace and global traces that follows
the same mapping $\map_i$.
\begin{figure}
    \begin{tikzpicture}[>=latex,font=\sffamily, node distance=45pt, scale=.8]
      \tikzstyle{ststate} = [circle, fill=pink!!10, draw]
      \tikzstyle{state} = [circle, fill=white, draw]
      \tikzstyle{textt} = [text width=5em]

        \node[state] (gs1) {\small $\bar{s}_{j+0}$};
        \node[ststate, right of = gs1] (gsst1) {\small $\bar{s}_{j+1}$};
        \node[ststate, right of = gsst1] (gsst2) {\small $\bar{s}_{j+2}$};
        \node[state, right of = gsst2] (gs2) {\small $\bar{s}_{j+3}$};
        \node[ststate, right of = gs2] (gsst3) {\small $\bar{s}_{j+4}$};
        \node[state, right of = gsst3] (gs3) {\small $\bar{s}_{j+5}$};

        \node[right = 15pt of gs3] (rdot2) {\small \dots};
        \node[left = 15pt of gs1] (ldot2) {\small \dots};
        \node[left = 10pt of ldot2] (globlabel) {\large $\pi$ };
        
        \node[textt, above =-1pt of gs1] (gas1) {\small $v_i,\neg v_o,\newline \run_i$};
        \node[textt, right of =gas1] (gasst1) {\small $\neg v_i, \neg v_o,\newline\neg \run_i$};
        \node[textt, right of =gasst1] (gasst2) {\small $\neg v_i, \neg v_o,\newline\neg \run_i$};
        \node[textt, right of = gasst2] (gas2) {\small $v_i,\neg v_o,\newline\run_i$};
        \node[textt, right of = gas2] (gasst3) {\small $v_i, v_o,\newline\neg \run_i$};
        \node[textt, right of = gasst3] (gas3) {\small $\neg v_i,v_o,\newline\run_i$};

      \node[state, below of = gs1] (s1) {\small $s_{i+0}$};
      \node[state, below of = gs2] (s2) {\small $s_{i+1}$};

      \node[state, below of = gs3] (s3) {\small $s_{i+2}$};

      \node[right = 15pt of s3] (rdot) {\small \dots};
      \node[left = 15pt of s1] (ldot) {\small \dots};
      \node[left = 10pt of ldot] (loclabel) {\large $\pi_i$ };
      
      \node[below =-1pt of s1] (as1) {\small $v_i,\neg v_o$};
      \node[below =-1pt of s2] (as2) {\small $v_i,\neg v_o$};
      \node[below =-1pt of s3] (as3) {\small $\neg v_i,v_o$};

      \draw[->] (ldot) edge[bend left] node[below]{} (s1);
      \draw[->] (s1) edge[bend left] node[below]{} (s2);

      \draw[->] (s2) edge[bend left] node[below]{} (s3);

      \draw[->] (s3) edge[bend left] node[below]{} (rdot);

        \draw[->] (ldot2) edge[bend left] node[below]{} (gs1);
        \draw[->] (gs1) edge[bend left] node[below]{} (gsst1);

        \draw[->] (gsst1) edge[bend left] node[below]{} (gsst2);

        \draw[->] (gsst2) edge[bend left] node[below]{} (gs2);

        \draw[->] (gs2) edge[bend left] node[below]{} (gsst3);

        \draw[->] (gsst3) edge[bend left] node[below]{} (gs3);
        \draw[->] (gs3) edge[bend left] node[below]{} (rdot2);

        \draw[->, color=red] (gs1.south) -- (s1.north);
        \draw[->, color=red] (gs2.south) -- (s2.north);
        \draw[->, color=red] (gs3.south) -- (s3.north);
    \end{tikzpicture}
  \caption{Graphical view of trace projection. White states of $\pi$ represent
  the states of the sequence $\map_i$, pink states represent states in which the
  local component stutters, and red arrows represent the link between the states
  of $\pi$ and the states of $\pi_i$ formally represented by $\map_k$.}
  \label{fig:pr}
\end{figure}

\begin{thm}
  \label{thm:proj}
  Let $\M = \M_1 \otimes \dots \otimes \M_n$:
  $$\text{For all } 1 \le i \le n: \La(\M_i) \supseteq \{ \pr_{\M_i}(\pi) | \pi \in \La(\M)\}$$
\end{thm}

\begin{proof}
  Given a trace $\pi \in \La(\M), \pi_i := \pr_{\M_i}(\pi)$.
  We prove the theorem by induction on the length of $\pi_i$.
  The inductive hypothesis states that if $|\pi_i| > k + 1$ and $\pi_i^{0 \dots k} \in \Laf(\M_i)$,
  then $\pi_i^{0 \dots k + 1} \in \Laf(\M_i)$.

  \begin{itemize}
    \item Base case:\\ 
      By definition $\I := \bigwedge_{1 \le i \le n} \I_i$ and $\pi_i(0) = \pi(\map_0)(\V_i)$.
      We derive that $\pi_i, 0 \models \I_i \Rightarrow \pi_i^{0 \dots 0} \in \Laf(\M_i)$.
      It should be noted that this holds also if $|\pi_i| = 1$ and $\pi, |\pi| - 1 \not\models run_i$
      because $\I_i$ is a proposition on symbols of $\Vo_i$.
    \item Inductive case:\\
      By definition $\pi_i(k) = \pi(\map_k)(\V_i)$ where $\map_k$ is the kth position
      s.t. $\run_i$ holds. Therefore, 
      $\pi(\map_k) \pi(\map_k + 1) \models \run_i \rightarrow \T_i \Rightarrow
      \pi(\map_k) \pi(\map_k + 1) \models \T_i$.
      By $\fcond^{\M_i}$ we obtain that $\pi_{\map_{k+1}}(\Vo_i) = \pi_{\map_k + 1}(\Vo_i) = \pi_i(k+1)(\Vo_i)$.
      Since $\T_i$ reasons over symbols of
      $\V \cup {\Vo}'$, then $\pi(\map_k) \pi(\map_{k} + 1) \models \T_i \Leftrightarrow
      \pi_i(k) \pi_i(k + 1) \models \T_i$. Therefore, $\pi_i, k \models \T_i$ which
      guarantees that $\pi_i^{0 \dots k + 1} \in \Laf(\M_i)$.
  \end{itemize}
  We proved the theorem for finite traces, we now extend the proof to consider infinite traces.
  To do so, it is sufficient to prove that each infinite trace $\pi_i$ satisfies
  the strong fairness conditions of $\SF_i$ since it already satisfies $\I_i$ and $\T_i$.

  Since $\pi_i = \pr_{\M_i}(\pi)$ and $\pi \in \La(\M)$, $\pi \models 
  \bigwedge_{\langle f_a, f_g \rangle \in \SF} (\always \future f_a \rightarrow \always \future 
  f_g)$. In particular, due to the composition,
  $\pi \models \bigwedge_{\langle f_a, f_g \rangle \in \SF_i}
  (\always \future (\run_i \wedge f_a) \rightarrow \always \future (\run_i \wedge f_g))$
  The projection defines the sequence $\map_0, \dots$ representing the points in which
  $\run_i$ holds; therefore, for all $\langle f_a, f_g\rangle \in $if for all $k$ there exists $j \ge k$ s.t. $\pi, \map_k \models f_a$
  then for all $k'$ there exists $j' \ge k'$ s.t. $\pi, \map_k' \models f_g$. From
  the projection definition then $\pi_i \models \always \future f_a \rightarrow \always \future f_g$
  for each $\langle f_a, f_g \rangle \in \SF_i$.
\end{proof}

\begin{defi}
  Let $\M_1, \dots, \M_n$ be $n$ \ita{}, we define $\gamma_S$ as follows:
  $$\gamma_S(\M_1, \dots, \M_n) := 
  \M_1 \otimes \dots \otimes \M_n$$
\end{defi}

From Theorem~\ref{thm:proj}, we obtain a composition $\gamma_S$ for which each projected global
trace is an actual behaviour of a local trace. Therefore, when we consider the global
traces, we do not introduce new local traces which cannot be witnessed locally.

To complete our compositional reasoning, we need to define the function $\gamma_P$
such that Inference \ref{ifr:comp} holds. To do so, Section \ref{sec:rewr} provides
a rewriting technique that maps each local trace satisfying $\varphi_i$ to a global
trace satisfying $\gamma_P$.

\section{Rewriting}
\label{sec:rewr}
In this section, we introduce a rewriting-based approach for the composition of
local properties. In section \ref{sec:r_simple}, we present the rewriting of
the logic of Section \ref{sec:logic} that maps local
properties to global properties with proofs and complexity results.
In section \ref{sec:trr}, we propose an optimized version of the rewriting.
Finally, in section \ref{sec:trrufa}, we propose
a variation of the optimized rewriting tailored for infinite executions of local
properties i.e. in which the semantics is the one of standard LTL because we assume
fairness of $\run_i$.

To simplify the notation, we assume to be given $n$ interface transition systems 
$\M_1, \dots,$ $\M_n$, a composed ITS $\mathcal{C} = \M_1 \otimes \dots \otimes \M_n$,
a trace $\pi$ of $\M_i$ with $i < n$, a local property $\varphi$ and a local term $u$. For brevity,
we refer to $\R_{\M_i}$, $\R^*_{\M_i}$, $\Ropt_{\M_i}$, $\St_{\M_i}$, $\pr_{\M_i}$,
$\e_i$ and $\run_i$ as respectively $\R, \R^*, \Ropt, \Ropt^*, \St, \pr, \e
\text{ and } \run$. Moreover, we will refer to $\map_0, \map_1, \dots$ as the sequence 
of definition \ref{def:pr}.
As in section \ref{sec:ltlsem}, we denote input predicates
with apex $I$: $\pred^I$ and
the output predicates and terms with apex $O$: $\pred^O$.
Finally, we denote $sgn \in \{ -, +\}$.

\subsection{Trucanted LTL with Event Freezing Function Compositional Rewriting}
\label{sec:r_simple}
In this section, we propose a rewriting to asynchronously
compose propositional truncated LTL properties over Interface Transition Systems symbols.

The idea is based on the notions of projection (introduced in Definition \ref{def:pr})
and asynchronous composition (introduced in Definition \ref{def:symbasynccomp}).
We produce a rewriting $\R^*$ for $\varphi$ such that each trace $\pi$ of $\M_i$
satisfies the rewritten formula iff the projected trace satisfies the original
property.

\begin{defi}
  We define $\R$ as the following rewriting function:
\begingroup
\allowdisplaybreaks
  \begin{align*}
    &\Rm(\pred^I(\term_1, \dots, \term_n)) := \neg \run \vee
    \pred^I(\Rm(\term_1), \dots, \Rm(\term_n))\\
    &\Rp(\pred^I(\term_1, \dots, \term_n)) := \run \wedge 
    \pred^I(\Rp(\term_1), \dots, \Rp(\term_n))\\
    &\Rs(\pred^O(\term_1, \dots, \term_n)) :=
      \pred^O(\Rs(\term_1), \dots, \Rs(\term_n))\\
    &\Rs(\varphi_1 \vee \varphi_2) := \Rs(\varphi_1) \vee \Rs(\varphi_2)\\
    &\Rm(\neg \varphi) := \neg \Rp(\varphi), \Rp(\neg \varphi) := \neg \Rm(\varphi)\\
    &\Rm(\nxt \varphi) := \nxt (\state \rel (\neg \state \vee \Rm(\varphi)))\\
    &\Rp(\nxt \varphi) := \nxt (\neg \state \unt (\state \wedge \Rp(\varphi)))\\
    &\Rm(\varphi_1 \unt \varphi_2) := (\neg \state \vee \Rm(\varphi_1)) \unt
      ((\state \wedge \Rm(\varphi_2)) \vee \y \e)\\
    &\Rp(\varphi_1 \unt \varphi_2) := (\neg \state \vee \Rp(\varphi_1)) \unt
      (\state \wedge \Rp(\varphi_2))\\
    &\Rs(\y \varphi) := \y (\neg \run \since (\run \wedge \Rs(\varphi)))\\
    &\Rs(\varphi_1\since \varphi_2) := (\neg \state \vee \Rs(\varphi_1)) \since
    (\state \wedge \Rs(\varphi_2))\\
    &\Rs(\fun(\term_1, ..., \term_n)) := \fun(\Rs(\term_1), ..., \Rs(\term_n))\\
    &\Rs(x) := x, \Rs(c) := c\\
    &\Rs(ite(\varphi, \term_1, \term_2)) :=
      ite(\Rp(\varphi), \Rm(\term_1), ite(\Rp(\neg \varphi), \Rm(\term_2), def_{ite(\varphi, u_1, u_2)}))\\
    &\Rs(next(u)) := \Rs(u) \atn (\state)\\
    &\Rs(u \atn \varphi) := \Rs(u) \atn (\state \wedge \Rp(\varphi))\\
    &\Rs(u \atl \varphi) := \Rs(u) \atl (\state \wedge \Rp(\varphi))
  \end{align*}

\endgroup
  where $\state := \run \vee (\z \run \wedge \e)$ and $\R(\varphi) = 
  \R^-(\varphi)$.
\end{defi}
\noindent 
$\R$ transforms the formula applying $\run$ and $\state$ to $\varphi$. The general
intuition is that when $\run$ is true, the local component triggers a transition.
Moreover, $\state$ represents a local state of $\pi$ in the global trace. It 
should be noted that the rewriting
has to deal with the last state. Therefore, $\state$ is represented by a state that either satisfies $\run$ or is the successor of the last state that
satisfies $\run$.

For output predicates, the rewriting is transparent because the projection definition
guarantees that each state $i$ of $\pi$ has the same evaluation of each $\map_i$
state of $\pist$. Although this holds for input predicates as well, there is a
semantic technicality that should be taken into account. Input predicates are
evaluated differently on the last state: with \textit{strong} semantics ($+$) they
do not hold while on \textit{weak} semantics ($-$) they hold.
Since (i) the last state of a local trace might not be the last state of the global
trace, and (ii) a variable that locally is an input variable might be an output
variable of the composed system, we have to take input predicates into account inside the
rewriting.
Therefore, with \textit{strong} semantics, the rewriting forces the predicate
to hold in a transition i.e. when $\run$ holds ($\pi, i \mts \pred^I \Leftrightarrow
i< |\pi| - 1$ and $\pred^I(\pi(i)(\term_1),\dots, \pi(i)(\term_n))$).
The \textit{weak} semantics is managed specularly.

Regarding $\nxt$, $\Rm$ needs to pass from point $\map_i$ to $\map_{i+1}$ and
to verify that the sub-formula is verified in that state. The first $\nxt$ passes
from $\map_i$ to $\map_i + 1$. Then, the rewriting skips all states that are not $\map_{i+1}$
and ``stops'' in position $\map_{i+1}$.
Figure \ref{fig:rewr_x} shows an intuitive representation of the rewriting of $\nxt$ when
the trace is not in its last state.

For Until formulae, the rewriting needs to ``skip'' all points that either do not
belong to the local trace or satisfy the left part of the formula while it must
``stop'' to a point that satisfies the right part and is a state in the local trace.
To do so, it introduces a disjunction with $\state$ on the left side of the
formula and a conjunction with $\state$ on the right side of the formula.
When the rewriting and the formula are interpreted weakly, reaching the end of the local
trace makes the formula true; therefore, $\y \e$ is also put in disjunction
with the right side of the formula.
Figure \ref{fig:rewr_u} shows an intuitive representation of the rewriting of $\unt$
when the trace does not terminate before reaching $b$.
\begin{figure}
  \begin{tikzpicture}[>=latex,font=\sffamily, node distance=50pt]
    \tikzstyle{ststate} = [circle, fill=pink!!10, draw]
    \tikzstyle{state} = [circle, fill=white, draw]

    \node[state] (s1) {\small $s_{i+0}$};
    \node[state, below= 40pt of s1] (gs1) {\small $\bar{s}_{j+0}$};
    \node[ststate, right of = gs1] (gsst1) {\small $\bar{s}_{j+1}$};
    \node[ststate, right of = gsst1] (gsst2) {\small $\bar{s}_{j+2}$};
    \node[state, right of = gsst2] (gs2) {\small $\bar{s}_{j+3}$};
    \node[state, right of = gs2] (gs3) {\small $\bar{s}_{j+4}$};

    \node[right of = gs3] (rdot2) {\small \dots};
    \node[left of = gs1] (ldot2) {\small \dots};
    \node[left = 10pt of ldot2] (globlabel) { $\pist$ };
    
    \node[above =-1pt of gs1] (gas1) {\small $a,\run$};
    \node[right of =gas1] (gasst1) {\small $\neg a, \neg \run$};
    \node[right of =gasst1] (gasst2) {\small $\neg a, \neg\run$};
    \node[right of = gasst2] (gas2) {\small $a, \run$};
    \node[right of = gas2] (gas3) {\small $\neg a, \run$};

    \draw[->] (ldot2) edge[bend left] node[below]{} (gs1);
    \draw[->] (gs1) edge[bend left] node[below]{} (gsst1);

    \draw[->] (gsst1) edge[bend left] node[below]{} (gsst2);

    \draw[->] (gsst2) edge[bend left] node[below]{} (gs2);

    \draw[->] (gs2) edge[bend left] node[below]{} (gs3);

    \draw[->] (gs3) edge[bend left] node[below]{} (rdot2);

    \node[state, above=40pt of gs2] (s2) {\small $s_{i+1}$};

    \node[state, above=40pt of gs3] (s3) {\small $s_{i+2}$};

    \node[right of = s3] (rdot) {\small \dots};
    \node[left of = s1] (ldot) {\small \dots};
    \node[left = 10pt of ldot] (loclabel) { $\pi$ };
    
    \node[above =-1pt of s1] (as1) {\small $a$};
    \node[above = -1pt of s2] (as2) {\small $a$};
    \node[above = -1pt of s3] (as3) {\small $\neg a$};

    \draw[->] (ldot) edge[bend left] node[below]{} (s1);
    \draw[->] (s1) edge[bend left] node[below]{} (s2);

    \draw[->] (s2) edge[bend left] node[below]{} (s3);

    \draw[->] (s3) edge[bend left] node[below]{} (rdot);

    \path (s2.south west)
         edge[decorate,decoration={brace,mirror,raise=.15cm},"\small$a$"below=6pt]
          (s2.south west -| s2.south east);

    \path (s1.south west)
         edge[decorate,decoration={brace,mirror,raise=.15cm},"\small$\nxt a$"below=6pt]
          (s1.south west -| s1.south east);

    \path (gsst1.south west)
         edge[decorate,decoration={brace,mirror,raise=.15cm},"\small$\neg \run \vee a$"below=6pt]
          (gsst1.south west -| gs2.south east);

    \path (gs2.south west)
          edge[decorate,decoration={brace,mirror,raise=.30cm},"\small$\run \wedge a$"below=10pt]
          (gs2.south west -| gs2.south east);
    \path (gsst1.south west)
          edge[decorate,decoration={brace,mirror,raise=.45cm},"\small$\state \rel (\neg \state \vee a)$"below=15pt]
          (gsst1.south west -| gsst1.south east);
    \path (gs1.south west)
          edge[decorate,decoration={brace,mirror,raise=.15cm},"\small$\Rm(\nxt a)$"below=6pt]
          (gs1.south west -| gs1.south east);
  \end{tikzpicture}
  \caption{Graphical representation of rewriting of $\nxt a$. $\pi$ represents the
  local trace while $\pist$ represents the trace of the composition. White states
  are states of local trace while pink states are states in which the local component is not running.
  In this example, $a$ is an input variable which happens to be true in state $s_i$
  and $s_{i+1}$ of local trace $\pi$ and state $\bar{s}_j$. To show the intuition
  of the rewriting, we show that Release operator permits to skip the pink states
  (which are not relevant w.r.t. the local trace). Finally, at state $\bar{s}_{j+3}$
  $a$ is evaluated since $\run$ is true.
  }
  \label{fig:rewr_x}
\end{figure}

\begin{figure}
  \begin{tikzpicture}[>=latex,font=\sffamily, node distance=50pt]
    \tikzstyle{ststate} = [circle, fill=pink!!10, draw]
    \tikzstyle{state} = [circle, fill=white, draw]

    \node[state] (s1) {\small $s_{i+0}$};
    \node[state, below= 50pt of s1] (gs1) {\small $\bar{s}_{j+0}$};
    \node[ststate, right of = gs1] (gsst1) {\small $\bar{s}_{j+1}$};
    \node[ststate, right of = gsst1] (gsst2) {\small $\bar{s}_{j+2}$};
    \node[state, right of = gsst2] (gs2) {\small $\bar{s}_{j+3}$};
    \node[state, right of = gs2] (gs3) {\small $\bar{s}_{j+4}$};

    \node[right of = gs3] (rdot2) {\small \dots};
    \node[left of = gs1] (ldot2) {\small \dots};
    \node[left = 10pt of ldot2] (globlabel) { $\pist$ };
    
    \node[above =-1pt of gs1] (gas1) {\small $a,\neg b,\run$};
    \node[right of =gas1] (gasst1) {\small $\neg a, b, \neg \run$};
    \node[right of =gasst1] (gasst2) {\small $\neg a, \neg b\neg\run$};
    \node[right of = gasst2] (gas2) {\small $a, \neg b, \run$};
    \node[right of = gas2] (gas3) {\small $\neg a, b, \run$};

    \draw[->] (ldot2) edge[bend left] node[below]{} (gs1);
    \draw[->] (gs1) edge[bend left] node[below]{} (gsst1);

    \draw[->] (gsst1) edge[bend left] node[below]{} (gsst2);

    \draw[->] (gsst2) edge[bend left] node[below]{} (gs2);

    \draw[->] (gs2) edge[bend left] node[below]{} (gs3);

    \draw[->] (gs3) edge[bend left] node[below]{} (rdot2);

    \node[state, above= 50pt of gs2] (s2) {\small $s_{i+1}$};

    \node[state, above= 50pt of gs3] (s3) {\small $s_{i+2}$};

    \node[right of = s3] (rdot) {\small \dots};
    \node[left of = s1] (ldot) {\small \dots};
    \node[left = 10pt of ldot] (loclabel) { $\pi$ };
    
    \node[above =-1pt of s1] (as1) {\small $a, \neg b$};
    \node[above = -1pt of s2] (as2) {\small $a, \neg b$};
    \node[above = -1pt of s3] (as3) {\small $\neg a, b$};

    \draw[->] (ldot) edge[bend left] node[below]{} (s1);
    \draw[->] (s1) edge[bend left] node[below]{} (s2);

    \draw[->] (s2) edge[bend left] node[below]{} (s3);

    \draw[->] (s3) edge[bend left] node[below]{} (rdot);

    \path (s3.south west)
         edge[decorate,decoration={brace,mirror,raise=.15cm},"\small$b$"below=6pt]
          (s3.south west -| s3.south east);

    \path (s1.south west)
         edge[decorate,decoration={brace,mirror,raise=.15cm},"\small$a$"below=6pt]
          (s2.south west -| s2.south east);

    \path (s1.south west)
         edge[decorate,decoration={brace,mirror,raise=.30cm},"\small$a \unt b$"below=10pt]
          (s1.south west -| s1.south east);
    \path (gs1.south west)
         edge[decorate,decoration={brace,mirror,raise=.15cm},"\small$\neg \state \vee a$"below=6pt]
          (gs1.south west -| gs2.south east);

    \path (gs3.south west)
          edge[decorate,decoration={brace,mirror,raise=.15cm},"\small$\state \wedge b$"below=6pt]
          (gs3.south west -| gs3.south east);
    \path (gs1.south west)
          edge[decorate,decoration={brace,mirror,raise=.30cm},"\small$\Rm(a \unt b)$"below=10pt]
          (gs1.south west -| gs1.south east);
  \end{tikzpicture}
  \caption{Graphical representation of rewriting of $\unt$. In this example both $a$ and $b$
  variables are input variables. Local trace $\pi$ satisfies $a\unt b$ at position $i$ since
  $a$ is true at state $s_i, s_{i+1}$ while $b$ is true at state $s_{i+2}$. 
  The corresponding global trace contains 2 additional states ($\bar{s}_{j+1},\bar{s}_{j+2}$)
  which are not considered in the rewriting since $\neg\state$ holds these 2 states.
  Finally, the state $\bar{s}_{j+4}$} of the global trace satisfies both $\state$ and $b$.
  \label{fig:rewr_u}
\end{figure}
\begin{lem}
  \label{thm:l1}
  For all $\pist \in \St(\pi)$, for all $i < |\pi|$:
  \begin{align*}
    &\pi, i \mts \varphi \Leftrightarrow \pist, \map_i \mts \R^{sgn}(\varphi)
  \end{align*}
\end{lem}
\begin{proof}
  We prove Lemma \ref{thm:l1} by induction on the formula. The high level intuition
  is that, assuming that we are in a point $\map_i$ of the trace, we can ``reach''
  the corresponding $\map_k$ using the syntactically rewriting for each temporal operator.
  Two graphical examples of that are given in Figure~\ref{fig:rewr_x} and Figure~\ref{fig:rewr_u}
  for respectively $\nxt$ and $\unt$.

  Before starting the proof we observe the following facts:
  \begin{enumerate}
    \item For all $i: i < |\pi| -1 \Leftrightarrow \pist, \map_i \mtp \run$.
      \label{itm:run}
    \item For all $i: i < |\pi| \Leftrightarrow \pist, \map_i \mtp \state$\label{itm:state}.
    \item $\forall i < |\pi|, \forall \map_{i-1} < j \le \map_i: \pist, \map_i \mtn \psi
      \Leftrightarrow \pist, j \mtn \state \rel (\neg\state \vee \psi)$\label{itm:weak}. 
    \item $\forall i < |\pi|, \forall \map_{i-1} < j \le \map_i: \pist, \map_i \mtp \psi
      \Leftrightarrow \pist, j \mtp \neg \state \unt (\state \wedge \psi)$\label{itm:str}. 
    \item $\pist, \map_{|\pi| - 1} + 1\mts \y \e$ \label{itm:ye}.
  \end{enumerate}
  Base cases:
  \begin{itemize}
    \item $\pred^I$: $\pi, i \mtn \pred^I \Leftrightarrow i \ge |\pi| -1 \vee \pred^I
      \overset{\pr}{\Leftrightarrow} i \ge |\pi| - 1 \text{ or }
      \pist, \map_i \mtn \pred^I \overset{\ref{itm:run}}{\Leftrightarrow} \pist, \map_i 
      \mtn \neg \run \vee \pred^I$. The strong semantics case is identical.
    \vskip0.25\baselineskip 
    \item $\pred^O, x, c$: Trivial.
  \end{itemize}
  Inductive cases:
  \begin{itemize}
    \item $\neg$: Trivial in both cases. It follows the semantics definition
    \item $\vee, \pred^I, \pred^O, \fun$: Trivial.
    \item $\nxt \varphi:
      \pist, \map_i \mtn \R(\nxt \varphi) \Leftrightarrow \pist, \map_i + 1 \mtn \state \rel 
      (\neg \state \vee \R(\varphi)) \overset{\ref{itm:weak}}{\Leftrightarrow} \pist,$ $\map_{i+1
      } \mtn \R(\varphi)$ if $i < |\pi| - 1$ (By  ind. holds). 
      Otherwise, $i = |\pi| - 1$ and thus $\pist, \map_i + 1 \mtn \always \neg \run$ which
      implies $\pist, \map_i \models \R(\nxt \varphi)$ as expected by the weak semantics.
      We skip the strong case because it is similar.
      \vskip0.25\baselineskip 
    \item $\varphi_1 \unt \varphi_2:
      \pist, \map_i \mtn \Rm(\varphi_1 \unt \varphi_2) \Leftrightarrow \pist, \map_i \mtn
      (\state \vee \Rm(\varphi_1)) \unt (\neg \state \wedge \Rm(\varphi_2) \vee \y \e) \Leftrightarrow
      \exists k' \ge \map_i \text{ s.t. } (\pist, k' \mtn \Rm(\varphi_2) \text{ and }\pist, k' \mtn
      \state \text{ or } \pist, k'$ $\mtn \y \e) \text{ and } \forall \map_i \le j' < k': \pist, j' \mtn \R(\varphi_1)
      \text{ or } \pist, j' \mtn \state \overset{\ref{itm:state},\ref{itm:ye}}{\Leftrightarrow}
      \exists k \ge i \text{ s.t. } \pist, \map_k$ $\mtn \R(\varphi_2) \text{ and } k < |\pi| \text{ or } k = |\pi| \text{ and }
      \forall i \le j < k \pist, \map_j \mtn \R(\varphi_1) \overset{Ind. (k,j < |\pi|), \forall \phi: \pi, |\pi| \mtn \phi}{\Leftrightarrow}
      \exists k \ge i \text{ s.t. } \pi, k \mtn \varphi_2 \text{ and }
      \forall i \le j < k: \pi, j \mtn \varphi_1 \Leftrightarrow \pi, i \mtn \varphi_1 \unt \varphi_2$. The proof of the strong case is the same.
    \vskip0.25\baselineskip 
    \item $\y \varphi$: The case is specular to the case of $\nxt$. In this specific
      case, we can use $\run$ instead of $\state$ because we are assuming that
      $i < |\pi|$; therefore $i-1 < |\pi| - 1 \overset{\ref{itm:run}}\Rightarrow \pist, \map_{i-1} \mtn\run$ (if $i = 0$, then property is false).
      \vskip0.25\baselineskip 
    \item $\varphi_1 \since \varphi_2:$ The case is specular to $\unt$. It is sufficient
      to observe that \ref{itm:str} can be applied for the past as well.
      \vskip0.25\baselineskip 
    \item $ite(\varphi, \term_1, \term_2):
      \pist(\map_i)(\Rs(ite(\varphi, \term_1, \term_2))) = \pist(\map_i)(\Rm(\term_1))
      \text{ if } \pist, \map_i \mtp \varphi; \dots = \pist(\map_i)(\Rm(\term_2)) 
      \text{ if } \pist, \map_i \mtp \neg \varphi; \dots = \pist(def_{ite(\varphi, u_1, u_2)}) 
      \text{ otherwise}$. By induction the case holds.
      \vskip0.25\baselineskip 
    \item $\term \atn \varphi:
      \pist(\map_i)(\Rs(\term \atn \varphi)) = \pist(\map_i)(\Rm(\term) \atn (\state \wedge \Rp(\varphi)))$.
      By the semantics of the at next operator operator, we obtain the following.

      $\pist(\map_i)(\Rs(\term \atn \varphi)) = \pist(j)(\Rm(\term))$ if there exists $j > \map_i \text{ s.t. }
      \pist, j \mtp \state \text{ and } \pist, j \mtp \Rp(\varphi)$; $\pist(\map_i)(\dots) = def_{\term \atn \varphi}$ otherwise.
      $\pist(\map_i)(\dots) \overset{\ref{itm:state}}{=} \pist(\map_j)(\Rm(\term)) \text{ if }
      \exists j \ge i \text{ s.t. } \pist, \map_j \mtp \Rp(\varphi); \pist(\map_i)(\dots) = def_{\dots}$ otherwise. Finally, by induction hypothesis, the case is proved.
      \vskip0.25\baselineskip 
    \item $\term \atl \varphi:$ Identical to $\atn$.
    \item $next(\term)$: Identical to case $\term \atn \top$. \qedhere
  \end{itemize}
\end{proof}
\noindent 
Lemma \ref{thm:l1} states that each point in the sequence $\map_0, \dots$ 
of $\pist$ satisfies $\R(\varphi)$ iff the local property satisfies in that
point $\varphi$; therefore, providing a mapping between the two properties.

\begin{defi}
  We define $\R^*$ as $ \R^*(\varphi) := \state \rel (\neg \state \vee \R(\varphi))$.
\end{defi}

\begin{lem}
  \label{thm:l2}
  For all $\pist \in \St(\pi):
  \pist, \map_0 \mt \R(\varphi) \Leftrightarrow
  \pist, 0 \mt \R^*(\varphi)$
\end{lem}

\begin{proof}
  Proofs follows simply applying observation~\ref{itm:weak} of the proof of Lemma~\ref{thm:l1}.\qedhere
\end{proof}

Lemma \ref{thm:l1} shows that $\R$ guarantees that satisfiability is preserved 
in the active transitions of the global traces. However,
$\map_0$ is not always granted to be equal to $0$ (see definition~\ref{def:pr}), and thus,
the rewriting must guarantee that satisfiability is preserved in the first transition as well.
From lemma \ref{thm:l1} and lemma \ref{thm:l2}, we infer the rewriting theorem (Theorem~\ref{thm:rewr_thm}) which shows that $\R^*$ maps the local trace to the global trace as follows:
\begin{thm}
  \label{thm:rewr_thm}
  For all $\pist \in \St(\pi):
  \pi \mt \varphi \Leftrightarrow
  \pist \mt \R^*(\varphi)$
\end{thm}

\begin{proof}
  We prove Theorem~\ref{thm:rewr_thm} through Lemma \ref{thm:l1} and Lemma \ref{thm:l2}:

  By Lemma \ref{thm:l1}, $\forall i: \pi, i \mts \varphi \Leftrightarrow \pist, \map_i \mts \Rs(\varphi)$. Therefore, $\pi \mtn \varphi \Leftrightarrow \pist, \map_0 \mtn \Rm(\varphi)$.
  By Lemma \ref{thm:l2}, $\pist, \map_0 \mts \Rm(\varphi) \Leftrightarrow \pist \mts \R^*(\varphi)$.
  Therefore, $\pi \mtn \varphi \Leftrightarrow \pist \mtn \R^*(\varphi)$.\qedhere
\end{proof}

%

\begin{thm}
  \label{thm:size}
  Let $\varphi$ be a truncated LTL formula:
  \begin{enumerate}
    \item If $\varphi$ does not contain $\ite$, the size of the rewritten formula
  with $\R^*$ is linear w.r.t. $\varphi$ i.e. $|\R^*(\varphi)| = O(|\varphi|)$.
    \item If $\varphi$ contains $\ite$, the size of the rewritten formula is in
      the worst-case exponentially larger than $\varphi$.
  \end{enumerate}
\end{thm}
\begin{proof}
  We prove Theorem~\ref{thm:size} by first showing that $\size{\R(\varphi)} = \size{\varphi} + c$ where $c$ is a constant inductively on the structure of the formula.
  The base case trivially holds since the $\Rs(x)=x$ and $\Rs(c)=c$. We now show
  the other cases assuming that the theorem holds on the sub-formulae. For brevity,
  we prove only the weak part of the rewriting; the prove of the strong part is
  identical since the sizes of the generated formulae are the same.
  \begin{itemize}
    \item ($\pred$) $\size{\Rm(\pred^I(\term_1,\dots, \term_n))} = 
      \size{\neg \run \vee \pred^I(\Rm(\term_1),\dots,\Rm(\term_n))} = 
      3 + \sum_{1 \le i \le n} \size{\term_i}$ $+ c_i = 2 + \sum_{1 \le i \le n} c_i +
      1 + \sum_{1 \le i \le n} \size{\term_i} = 2 + \sum_{1 \le i \le n} c_i + \size{\pred^I}(\term_1, \dots, \pred_n)$. Since $c_i$ are constants the rewriting is linear
      in this case. The proof for output predicate is almost identical.
    \item ($\vee,\neg$) Trivial.
    \item ($\nxt$) $\size{\Rm(\nxt \varphi)} = \size{\nxt (\state \rel (\neg \state
      \vee \Rm(\varphi)))} = 1 + 2 \size{\state} + 2 + \size{\Rm(\varphi)} =
      3 + \size{\state} + \size{\varphi} + c$. $\size{\state}$ and $c$ are constants;
      therefore, the rewriting is still linear.
      \vskip0.25\baselineskip 
    \item ($\unt$) $\size{\Rm(\varphi_1 \unt \varphi_2)} = \size{(\neg \state \vee
      \Rm(\varphi_1)) \unt (\state \wedge \Rm(\varphi_2))} = 
      3 + 2\size{\state} + 1 + \size{\varphi_1} + c_1 + \size{\varphi_2} + c_2 =
      \size{\varphi_1 \unt \varphi_2} + c_1 + c_2 + 3 + 2\size{\state}$. $\size{\state}$,
      $c_1$ and $c_2$ are constants; therefore, the rewriting is still linear.
    \item ($\since,\y$) The proof is respectively as $\unt$ and $\nxt$.
    \item ($\fun$) Trivial.
    \item ($\atn$) $\size{\Rs(\term \atn \varphi)} = \size{\Rs(\term) \atn (\state \wedge
      \Rp(\varphi))} = \size{\term} + c_1 + 1 + \size{\state} + 1 + \size{\varphi} + c_2
      = \size{\term \atn \varphi} + c_1 + c_2 + \size{\state} + 1$. Since $c_1$ (constant of $\Rs(\term)$), $c_2$ (constant of $\Rs(\varphi)$) and $\state$ are constants; then
      the rewriting is linear.
    \item ($\atl$, $next$) Identical to $\atn$.
  \end{itemize}
  Finally, $\size{\R^*(\varphi)} = \size{\state \rel (\neg \state \vee \Rm(\varphi))}
  = \size{\Rm(\varphi)} + 2\size{\state} + 3$; since $\size{\Rm(\varphi)}$ is linear
  w.r.t $\size{\varphi}$ we deduce that $\size{\R(\varphi)}$ is linear as well.

Finally, the size of the rewriting is worst-case exponential
when considering $\ite$ because $\Rs(\ite(\varphi, \term_1, \term_2))$ contains,
  both weakly and strongly $\R(\varphi)$, the rewriting of $\varphi$.\qedhere
\end{proof}

\begin{example}
  Consider the formula $\varphi_{c2} := \always (rec_2 \rightarrow out_2' = in_2 \wedge \nxt send_2)$ from Figure \ref{fig:ex}.
  The rewriting $\R^*(\varphi_{c2})$ is defined as the following formula.
  $$\always (\neg \state \vee (\underbrace{\run \wedge rec_2}_{\Rp(rec_2)} \rightarrow
  \overbrace{(out_2 \atn state = in_2)}^{\Rm(next(out_2))} \wedge \underbrace{\nxt (\run \rel (\neg \run \vee send_2))}_{\Rm(\nxt send_2)}))$$
  Finally, $\R^*(\varphi_{c2})$ is defined as $state \rel (\neg state \vee \Rm(\varphi_{c2})$.

  Recall that $\always$ is an abbreviation of Until: $\always \psi := \neg (\top \unt \neg\psi)$.
  Therefore, the rewriting of $\always$ is equal to $\neg\Rp(\top \unt \neg\psi) :=
  \neg ((\neg \state \vee \top) \unt (\state \wedge \Rp(\neg\psi)))$; we can simplify
  the left side of until and we obtain $\neg (\top \unt (\state \wedge \Rp(\neg \psi)))$; and,
  with further simplifications we obtain $\neg \future (\state \wedge \neg \Rm(\psi)) \equiv
  \always (\neg \state \vee \Rm(\psi))$. Intuitively, with the always modality, we
  are interested in evaluating the states that are local by evaluating as true
  any state in which the component stutters.

  For what regards $rec_2$, since it is in the left side of an implication, we rewrite it
  with strong semantics by asking for $\run$ to be true. Then, $next(out_2)$ is rewritten
  via at-next operator; the intuition is that at-next provides the ``next'' value
  of the variable $out_2$ if $\state$ will be true, otherwise a default value is given.
  For what regards next, the intuition is given by Figure~\ref{fig:rewr_x}.
\end{example}

\subsection{Optimized LTL compositional rewriting}
\label{sec:trr}
The main weakness of the rewriting proposed in previous sections is the size of
the resulting formula. However, there are several cases in which it is possible
to apply a simpler rewriting. For instance, $\always v_o$ is rewritten by $\Rs$
into $\always (\neg \state \vee v_o)$ while by $\fcond$ (see Definition~\ref{def:symbasynccomp}) 
it does not need to be rewritten since output variables do not change when $\run$ is false.
Similarly, $\nxt v_o$ can be rewritten in the weak semantics to $\e \vee \nxt v_o$.

To do so, we apply the concept of \textit{stutter-tolerance} introduced in \cite{nfm22} tailored
for possibly finite traces. Informally, a formula is said \textit{stutter-tolerant}
if it keeps the same value when rewritten with $\Rs$ in all adjacent stuttering transitions.

\begin{defi}
  \label{def:sttol}
  An LTL formula $\varphi$ and a term $u$ are respectively said \textit{stutter-tolerant} 
  w.r.t. $\Rs$ iff:\\
    $\text{For all } \pi, \text{for all } \pist \in \St(\pi), \text{for all } 0 \le i < |\pi|: 
    \text{for all }{\map_{i-1} < j < \map_i}:$
    \begin{align*}
      \pist, j \mts \Rs(\varphi) \Leftrightarrow& \pist, \map_i \mts \Rs(\varphi) \text{ and }\\
      \pist(j)(\R^{sgn}(u)) =& \pist(\map_i)(\R^{sgn}(u))
    \end{align*}
\end{defi}

\begin{defi}
  \label{def:syntst}
  An LTL formula $\varphi_{st}$ is \textit{syntactically stutter-tolerant}---abbreviated as synt.st.tol.---iff it has the following grammar:
  \begin{align*}
    \varphi_{st} :=&\varphi_{st} \vee \varphi_{st} \mid \neg \varphi_{st} \mid
  \pred^O(u, \dots, u)  \mid \varphi \unt \varphi \mid
  \y \varphi\\
    u_{st} := &\fun(u_{st}, \dots, u_{st}) \mid s \mid c \mid ite(\varphi_{st}, u_{\st}, u_{st})
  \mid u \atl \varphi
  \end{align*}
  where $\varphi$ is an LTL formula, $u$ is a term from LTL syntax, $s$ is an
  output variable and $c$ is a constant.
\end{defi}

\begin{lem}
  \label{thm:syntst}
  Syntactically stutter-tolerant formulas are stutter-tolerant w.r.t. $\Rs$
\end{lem}
\begin{proof}
  We prove the Lemma by induction on the size of the formula.
  
  The base case is trivial since output variables remain unchanged during stuttering.
  The inductive case is proved as follows:
  \begin{itemize}
    \item $\vee, \pred^O$ and $\neg$: Trivial.
    \item $\unt$:
      We can prove the correctness by induction on $j$, with base case $j = \map_i - 1$.
      $\pist, j \mtn (\neg \state \vee \Rm(\varphi_1)) \unt (\state \wedge \Rm(\varphi_2) \vee \e) \Leftrightarrow
      \pist, j \mtn (\Rm(\varphi_2) \wedge \state \vee \e) \vee (\neg \state \vee \Rm(\varphi_1)) \wedge \nxt (
      (\neg \state \vee \Rm(\varphi_1) \unt (\state \wedge \Rm(\varphi_2))) \overset{\pist, j \notmt \state}{\Leftrightarrow}
      \pist, j \mtn \nxt ((\Rm(\varphi_1) \vee \neg \state) \unt (\state \wedge \Rm(\varphi_2) \vee \e) \overset{j < |\pist| - 1}{\Leftrightarrow} \pist, j + 1 \mtn \Rm(\varphi_1 \unt \varphi_2)$
      which is $\map_i$ for the base case. The inductive case follows trivially.
    \item $\y$: The case of $\y$ can be prove in the same way of $\unt$ by expanding $\since$ instead of $\unt$.
    \item $\atl$: The semantics of at last evaluates $\Rp(\varphi)$ at the first occurrence in
      the past of $\state$. The proof is the same as $\y$.\qedhere
  \end{itemize}
\end{proof}
\noindent 
From Lemma~\ref{thm:syntst}, we derive a syntactical way to determine identify
a relevant fragment of stutter tolerant formulae. Since this definition is purely
syntactical, it is very simple for an algorithm to determine whether or not a
formula is syntactically stutter tolerant; it is sufficient to traverse the structure
of the formula and look at the variables and operators.

From the notion of syntactically stutter tolerant formula, we provide an optimized rewriting
that is semantically equivalent to $\Rs$. If the sub-formulas of $\varphi$
are syntactically stutter tolerant, we can simplify the rewriting for $\varphi$.

\begin{defi}
  \label{def:optrewr}
  We define $\Ropt$ as follows. We omit the cases that are identical to $\R$.
  \begin{align*}
    &\Roptm(\nxt \varphi) := \e \vee \nxt \Roptm(\varphi) \text{ If }
    \varphi \text{ is syntactically stutter-tolerant}\\
    &\Roptp(\nxt \varphi) :=
        \neg \e \wedge \nxt \Roptp(\varphi) \text{ If } \varphi \text{ is syntactically stutter-tolerant}\\
    &\Roptm(\varphi_1 \unt \varphi_2) :=
        \Roptm(\varphi_1) \unt (\y \e \vee \Roptm(\varphi_2)) \text{ If } \varphi_1,
        \varphi_2 \text{ are syntactically stutter-tolerant}\\
    &\Roptp(\varphi_1 \unt \varphi_2) :=
        \Roptp(\varphi_1) \unt (\neg \y \e \wedge \Roptp(\varphi_2)) \text{ If } \varphi_1,
        \varphi_2 \text{ are syntactically stutter-tolerant}\\
    &\Ropts(\term \atn \varphi) :=
          \Ropts(\term) \atn (\Ropts(\varphi) \wedge \neg \e) \text{ If } \varphi \text{ is syntactically stutter-tolerant}
  \end{align*}

  where $\Ropt(\varphi) := \Roptm(\varphi)$.
\end{defi}

\begin{lem}
  \label{thm:l1o}
  For all $\pi$, for all $\pist \in \St(\pi)$, for all $i < |\pi|$:
  \begin{align*}
    &\pi, i \mts \varphi \Leftrightarrow
    \pist, \map_i \mts \Ropt(\varphi)
    &\pi(i)(u) = \pist(\map_i)(\Ropt(u))
  \end{align*}
\end{lem}

\begin{proof}
  From Lemma \ref{thm:l1}, we deduce that to prove the Lemma it suffices to prove that $\forall_{0 \le i < |\pi|} \pist, \map_i \mts
  \R(\varphi) \Leftrightarrow \pist, \map_i \mts \Ropt(\varphi)$ and $\pist(\map_i)(\Rs(u)) = \pist(\map_i)(\Ropts(u))$.
  To prove that, we also prove inductively that 
  if $\varphi$ is syntactically stutter-tolerant, then $\forall_{0 \le i < |\pi|} \forall_{\map_{i-1} < j < \map_i}
  \pist, j \mts \Ropts(\varphi) \Leftrightarrow \pist, \map_i \Rs(\varphi)$.
\vskip0.5\baselineskip 
\noindent 
We need to prove only the cases in which the sub-formulae are stutter tolerant w.r.t
$\R^*$.

\begin{itemize}
  \item $\nxt$: If $\e$ is true, then $i >= |\pi| - 1$. Therefore, with weak semantics, the formula shall be true while with strong semantics the formula shall be false.

    If $\e$ is false, then $\pist, \map_i \mts \Ropts(\nxt \varphi) \Leftrightarrow$
    $\pist, \map_i \mts \nxt \Ropts(\varphi) \Leftrightarrow \pist, \map_i + 1 \mts \Ropts(\varphi)$.
    By induction hypothesis $\dots \Leftrightarrow \pist, \map_i + 1 \mts \R(\varphi) \Leftrightarrow \pist, \map_{i+1} \mts \R(\varphi)$.
  \item $\unt$: 
    $\pist, \map_i \mtn \Roptm(\varphi_1) \unt (\y \e \vee \Roptm(\varphi_2)) \Leftrightarrow
     \exists k' \ge \map_i$ s.t. $\pist, k' \mtn \y \e$ or $\pist, k' \mtn 
     \Roptm(\varphi_2)$ and $\forall_{\map_i \le j' < k'} \pist, j' \mtn \Roptm(\varphi_1)$.
     By induction hypothesis $\exists \map_k \ge \map_i \text{ s.t. } \pist, \map_k \mtn
     \R(\varphi_2) \text{ and } \forall_{\map_i \le \map_j < \map_k} \pist, \map_j \mtn \R(\varphi_1)$.

     We now prove the additional part of the Lemma.
     If $\varphi_1$ or $\varphi_2$ are not syntactically stutter tolerant,
     then the rewriting is identical and thus, the lemma holds since $\varphi$ is
     stutter tolerant.

     If $\varphi_1$ and $\varphi_2$ are syntactic stutter tolerant,
     $\pist, j \mtn \Roptm(\varphi_1) \unt (\y \e \vee \Roptm(\varphi_2))
     \Leftrightarrow \exists k' \ge \map_i \text{ s.t. } \pist, k' \mtn \y \e \text{ or } \pist, k' \mtn 
     \Roptm(\varphi_2)$ and $\forall_{\map_i \le j' < k'} \pist, j' \mtn \Roptm(\varphi_1)$.
     We observe that $\pist, k' \mtn \y \e \Leftrightarrow k' > \map_{|\pi| - 1}$.
     Moreover, by induction hypothesis, each $j'$ and each $k'$ can be replaced with
     respectively $\map_j$ and $\map_k$
     s.t. $i \le k \le |\pi|$ and $i \le j < k$. Therefore, $\dots \Leftrightarrow
     \exists k \ge i \text{ s.t. } \pist, \map_k \mtn \R(\varphi_2)$ or $k = |\pi|$
     and $\forall i \le j < k: \pist, \map_j \mtn \R(\varphi_1) \Leftrightarrow \pi, i \mtn \varphi_1 \unt \varphi_2$.
   \item $\atn$: It follows from induction hypothesis and by Lemma \ref{thm:l1}\qedhere
\end{itemize}
\end{proof}
\begin{defi}
  We define ${\Ropt}^*$ as follows:
  $$
  {\Ropt}^*(\varphi) :=
  \begin{cases}
    \Roptm(\varphi) & \text{ If }\varphi \text{ is syntactically stutter-tolerant}\\
    \state \rel (\neg \state \vee \Roptm(\varphi)) & \text{ Otherwise}
  \end{cases}$$
\end{defi}

\begin{lem}
  \label{thm:l2o}
  For all $\pist \in \St(\pi):
  \pist, \map_0 \mt \Ropt(\varphi) \Leftrightarrow
  \pist, 0 \mt \Ropt^*(\varphi)$
\end{lem}
\begin{proof}
  If $\varphi$ is not syntactically stutter tolerant, then the proof is the same
  one of lemma \ref{thm:l2}.
  If $\varphi$ is syntactically stutter tolerant, then $\forall_{\map_{-1} < j < \map_i}
  \pist, j \mtn \Roptm(\varphi) \Leftrightarrow \pist, \map_0 \mtn \Roptm(\varphi)$.
  Since $\map_{-1}$ is $-1$, then either $\map_0=0$ or $j$ gets the value $0$ in the
  for all; thus, proving the lemma.\qedhere
\end{proof}

\begin{thm}
  \label{thm:opt_thm}
  For all $\pist \in \St(\pi): \pi \mt \varphi \Leftrightarrow
  \pist \mt \Ropt^*(\varphi)$
\end{thm}

\begin{proof}
  The proof is identical to the proof of Theorem \ref{thm:rewr_thm} using
  Lemma \ref{thm:l1o} and Lemma \ref{thm:l2o}.\qedhere
\end{proof}
\noindent 
Theorem ~\ref{thm:rewr_thm} and Theorem \ref{thm:opt_thm} show that respectively
$\R^*$ and $\Ropt^*$ are able to translate a local LTL
property into a global property without changing its semantics in terms of traces.
Therefore, we can use the two rewritings to prove Inference \ref{ifr:comp}.
\begin{defi}
  \label{def:gamma_p}
  Let $\M_1, \dots, \M_n$ be $n$ ITS and $\Prop_1, \dots, \Prop_n$ be LTL formulas
  on the language of each $\M_i$. We define $\gamma_P$ as follows
  $$\gamma_P(\Prop_1, \dots, \Prop_n) :=
  \Ropt^*_{\M_1}(\Prop_1) \wedge \dots \wedge \Ropt^*_{\M_n}(\Prop_n) \wedge \fcond$$
  where $\fcond := \fcond^{\M_1} \wedge \dots \wedge
  \fcond^{\M_n}$
\end{defi}
\begin{cor}
  \label{crl:ifrok}
  Using $\gamma_P$ from Definition \ref{def:gamma_p}, $\gamma_S$ from Section \ref{sec:mainp},
  for all compatible ITS $\M_1, \dots, \M_n$, for all local properties $\Prop_1, \dots, \Prop_n$
  over the language of respectively $\M_1, \dots, \M_n$, for all global properties $\Prop$:
   Inference \ref{ifr:comp} holds.
\end{cor}

\begin{example}
  Consider the formula $\varphi_{c2} := \always (rec_2 \rightarrow out_2' = in_2 \wedge \nxt send_2)$ from Figure \ref{fig:ex}.
  The rewriting ${\Ropt}(\varphi_{c2})$ is defined as the following formula.
  $$\always (\neg \state \vee (\underbrace{\run \wedge rec_2}_{\Roptp(rec_2)} \rightarrow
  \underbrace{(\e \vee \overbrace{out_2 \atn \neg \e}^{\Ropt(out_2 \atn \top)}=in_2)}_{\Ropt(out_2'=in_2)} \wedge
  \underbrace{\e \vee \nxt send_2}_{\Roptm(\nxt send_2)}))$$
  Finally, ${\Ropt}^*(\varphi_{c2})$ is defined as $\Ropt(\varphi_{c2})$.

  The simplification optimizes the formula in various parts. $\nxt send_2$ is rewritten
  into a simpler formula, in which we only need to check whether we are at the end
  of the local trace; the same occurs for primed output variable ($next$).
  On the contrary, $\always$ must be rewritten as for $\R$ because its sub-formula
  is not stutter-tolerant. On the other hand, the top-level rewriting is simplified
  because $\always$ is syntactically stutter tolerant.
\end{example}
\subsection{Rewriting under fairness assumption}
\label{sec:trrufa}
This section defines a variation of the previous rewriting that assumes infinite
execution of local components. The rewriting is presented as a variation of
the optimized rewriting; it is meant to exploit the fairness assumption to be more concise
and efficient. The general idea is that since the local components run infinitely
often, we can consider the semantics of LTL with event-freezing functions instead
of the finite semantics considered up to now.

\begin{defi}
  \label{def:optrewrfair}
  We define $\Td$ as follows:
  \begin{align*}
    &\Td(v) := v \text{ for } v \in \V\\
    &\Td(\pred(\term_1,\dots, \term_n)) := \pred(\Td(\term_1), \dots,\Td(\term_n))\\
    &\Td(\varphi \vee \psi) := \Td(\varphi) \vee \Td(\psi)\\
    &\Td(\neg \varphi) = \neg \Td(\varphi)\\
    &\Td(X \psi) :=
        \begin{cases}
          \nxt (\Td(\psi)) & \text{if } \psi \text{ is synt.st.tol.}\\
          \nxt (\run \rel (\neg \run \vee \Td(\psi))) & \text{otherwise}
        \end{cases}\\
     &\Td(\varphi_1 \unt \varphi_2) :=
      \begin{cases}
        \Td(\varphi_1) \unt \Td(\varphi_2) & \text{if } \varphi_1 \text{ and } \varphi_2 \text{ are synt.st.tol.}\\
        (\neg \run \vee \Td(\varphi_1)) \unt (\run \wedge \Td(\varphi_2)) & \text{otherwise}
      \end{cases}\\
     &\Td(\y \varphi) := \y (\neg\run \since (\run \wedge \Td(\varphi)))\\
     &\Td(\varphi_1 \since \varphi_2) :=
      \begin{cases}
        \Td(\varphi_1) \since \Td(\varphi_2) & \text{if } \varphi_1 \text{ and } \varphi_2 \text{ are synt.st.tol.}\\
        (\neg \run \vee \Td(\varphi_1)) \since (\run \wedge \Td(\varphi_2)) & \text{otherwise}
      \end{cases}\\
     &\Td(ite(\psi, \term_1, \term_2)) := ite(\Td(\psi), \Td(\term_1), \Td(\term_2))\\
     &\Td(\term \atn \varphi) :=
      \begin{cases}
        \Td(\term) \atn \Td(\varphi) & \text{if } \psi \text{ is synt.st.tol.}\\
        \Td(\term) \atn (\run \wedge \Td(\varphi)) & \text{otherwise}
      \end{cases}\\
     &\Td(\term \atl \varphi) :=\Td(\term) \atl (\run \wedge \Td(\varphi))
  \end{align*}

  Moreover, we define ${\Td}^*$ as $$ {\Td}^*(\varphi) := \begin{cases}
   \Td(\varphi)& \text{ If } \varphi \text{ is syntactically stutter-tolerant}\\
   \run \rel (\neg \run \vee \Td(\varphi))& \text{ Otherwise}
  \end{cases}$$
\end{defi}
\noindent 
${\Td}^*$ simplifies $\Ropt^*$ in various ways. It does not need to distinguish between
input and output variables, it does not distinguish between weak/strong semantics
and it does not distinguish between $state$ and $\run$. Intuitively, these distinctions
make sense only with finite semantics; since the rewriting assumes infinite local
executions, these technicalities become superfluous.

\begin{thm}
  Let $\pi$ be an \textit{infinite} trace, for all $\pist \in \St(\pi):
  \pi \models \varphi \Leftrightarrow \pist \models {\Td}^*(\varphi)$
\end{thm}

\begin{proof}{(Sketch)}
  Since $\pi$ is infinite, $\e$ is always false, $state \Leftrightarrow \run$.
  We can substitute $\e$ with $\bot$ and $state$ with $\run$ in $\Ropt$ and we
  obtain this rewriting.\qedhere
\end{proof}

\begin{example}
  Consider the formula $\varphi_{c2} := \always (rec_2 \rightarrow out_2' = in_2 \wedge \nxt send_2)$ from Figure \ref{fig:ex}.
  The rewriting ${\Td}(\varphi_{c2})$ is defined as the following formula.
  $$\always (\neg \run \vee (\underbrace{rec_2}_{\Td(rec_2)} \rightarrow
  \underbrace{(\overbrace{out_2'}^{\Td(out_2')}=in_2)}_{\Td(out_2'=in_2)} \wedge
  \underbrace{\nxt send_2}_{\Td(\nxt send_2)}))$$
  Finally, ${\Td}^*(\varphi_{c2})$ is defined as $\Td(\varphi_{c2})$.

By assuming infinite executions of local traces it is possible to drastically simplify
the rewriting. Next outputs are left unchanged because they are stutter tolerant
and with infinite execution $\e$ is not needed. Input variables are also transparent
to the rewriting because the infinite semantics do not distinguish between input
and output variables.
As before, $\always$ must be rewritten because its sub-formula
is not stutter-tolerant; however, in this case, the rewriting can use $\run$ instead
of $state$ because the two expressions are equivalent with infinite semantics.
Finally, as before, the top-level rewriting is simplified because $\always$ is syntactically stutter tolerant.
\end{example}

\section{Experimental evaluation }
\label{sec:perfeval}
We implemented the compositional techniques proposed in this paper inside the
contract-based design tool OCRA~\cite{ocra}. Our extension covers the contract 
refinement check, in which a contract (defined as a couple $\langle A, G \rangle$
of LTL formulae) is considered correct if its sub-component contracts
refine it. For simplicity, we consider contracts without assumptions i.e. in which
the assumption is $\top$; in this scenario, the refinement of contracts is given
by the compositional reasoning described in this paper. For a detailed description
of the assume-guarantee contract proof system employed in the tool refer to \cite{CT15}.

The objective of this experimental evaluation is to do both a quantitative and
a qualitative evaluation of our compositional approach differentiating possibly 
finite and infinite semantics.
Qualitatively, we compare the verification results to assess how the finiteness
impacts on the result, observing also the required scheduling assumptions for possibly
finite systems.
Quantitatively, we analyse the overhead of reasoning over a mix of finite and
infinite executions to assess whether or not our rewriting scales on real models.
Due to the absence of equally expressible formalism to define asynchronous composition,
our comparison with related work is limited to another rewriting
suited only for local infinite systems with event-based asynchronous composition\cite{oldrewr}.

Although we defined compositional reasoning with truncated semantics, in our
experiments we reason about global infinite executions and possibly finite local
executions.
The experiments~\cite{zenodoexp} were run in parallel on a cluster with nodes
with Intel Xeon CPU 6226R running at 2.9GHz with 32CPU, 12GB. The timeout for
each run was two hours and the memory cap was set to 2GB.

In this section, we denote the composition based on $\Ropt^*$ as \trr{}(Truncated Rewriting); we denote
the composition based on $\Ropt^*$ with additional constraints to ensure infinite
local executions as \trrf{} (Truncated Rewriting + Fairness); we denote the
composition based on ${\Td}^*$ assuming local infinite execution as \trrufa{} 
(Truncated Rewriting under Fairness Assumption).

\begin{figure}[t]
  \begin{subfigure}{.49\textwidth}
    \includegraphics[width=\textwidth]{"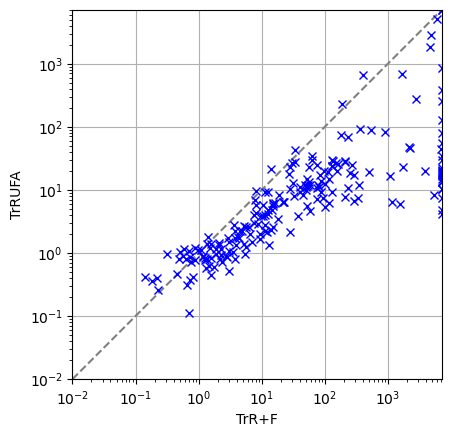"}
    \caption{Comparison between \trrf{} and \trrufa{} over valid instances.}
    \label{fig:compfair}
  \end{subfigure}
  \begin{subfigure}{.49\textwidth}
    \includegraphics[width=\textwidth]{"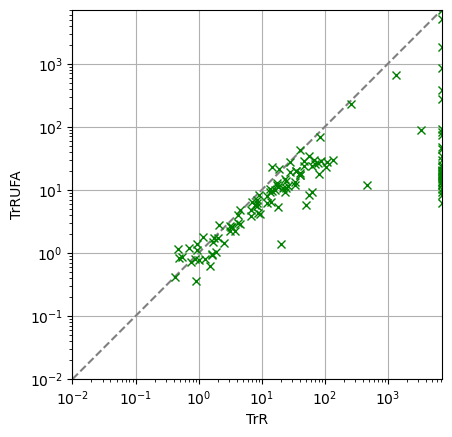"}
    \caption{Comparison between \trr{} and \trrufa{} over valid instances.}
  \end{subfigure}
  \begin{subfigure}{.49\textwidth}
    \includegraphics[width=\textwidth]{"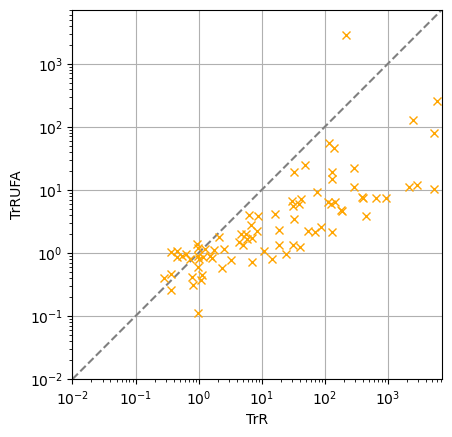"}
    \caption{Comparison between \trr{} and \trrufa{} over "different"
    results.}
  \end{subfigure}
  \begin{subfigure}{.49\textwidth}
    \includegraphics[width=\textwidth]{"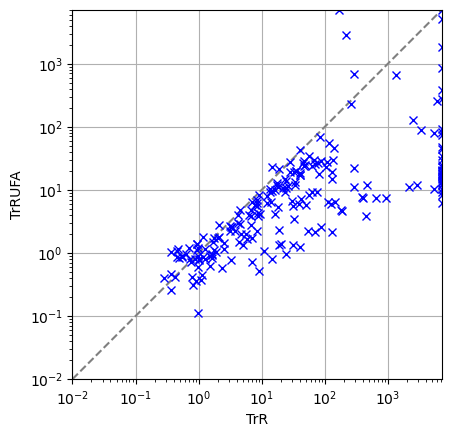"}
    \caption{Comparison between \trr{} and \trrufa{} over all the instances.}
  \end{subfigure}
  \caption{Scatter plots comparing \trr{},\trrf{} and \trrufa}
  \label{fig:comp1}
\end{figure}

\begin{figure}
  \includegraphics[width=.49\textwidth]{"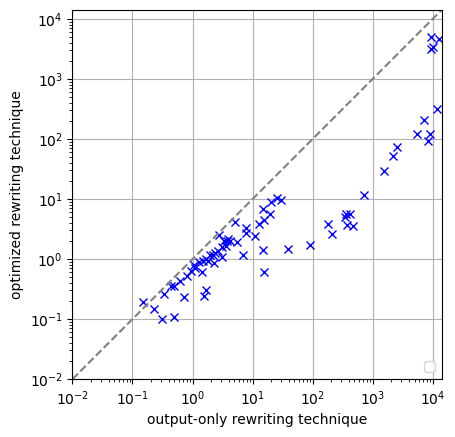"}
  \caption{Comparison between \trrufa{} and event-based rewriting
  of \cite{oldrewr}}
  \label{fig:old_comp}
\end{figure}

\begin{table}
\begin{filecontents*}{table.csv}
Model,Algorithm,Result,Time (s)
Response boolean sequence with size=3,TrR,Invalid,0.74
Response boolean sequence with size=3,TrRuFA,Valid,0.79
Response boolean sequence with size=3,TrR+F,Valid,0.48
Simplified Example (no component $c_3$),TrR,Invalid,1.15
Simplified Example (no component $c_3$),TrRuFA,Valid,0.87
Simplified Example (no component $c_3$),TrR+F,Valid,2.82
Example (Figure \ref{fig:ex}), TrR,Valid,20.13
Example (Figure \ref{fig:ex}), TrR+F,Valid,5.3
Example (Figure \ref{fig:ex}), TrRuFA,Valid,1.37
Response event sequence with size=12,TrR,Invalid,14.61
Response event sequence with size=12,TrRuFA,Valid,0.8
Response event sequence with size=12,TrR+F,Valid,3.74
Universality sequence with size=17,TrR,Invalid,6.63
Universality sequence with size=17,TrRuFA,Valid,2.81
Universality sequence with size=17,TrR+F,Valid,11.17
"Brake when AEB status is active" bugged,TrR,Invalid,48.3
"Brake when AEB status is active" bugged,TrRuFA,Valid,24.63
"Brake when AEB status is active" bugged,TrR+F,Valid,74.72

"Brake when AEB status is active" fixed, TrR, Valid,136.34
"Brake when AEB status is active" fixed, TrRuFA,Valid,17.6
"Brake when AEB status is active" fixed, TrR+F,Valid,62.83
\end{filecontents*}
\csvautobooktabular{table.csv}
  \caption{Subset of quantitative and qualitative results of the experimental evaluation.}
  \label{tbl:res_t}
\end{table}
\subsection{Benchmarks}
We have considered benchmarks of various kinds:
\begin{enumerate}
  \item Simple asynchronous models from OCRA comprehending the example depicted
    in Figure \ref{fig:ex}.
  \item Pattern formula compositions from \cite{nfm22} experimental evaluation
  \item Contracts from the experimental evaluation of \cite{CimattiTACAS23} on
    AUTOSAR models adapted for truncated semantics.
\end{enumerate}

\subsubsection{Simple asynchronous models}
We considered two variations of the example model depicted in Figure \ref{fig:ex}:
one version considers all three components with both the possible outcome
of failure and success. The second version is composed only of components $c1$ and $c2$ and the global property states that eventually the message is sent.

The second model we are considering is a simple representation of a Wheel Brake
System with two Braking System Control Units (BSCU) connected to two input braking
pedals, a Selection Switch and an actual hydraulic component that performs the actual brake. The top-level property
states that if one of the two pedals is pressed eventually the hydraulic component
will brake.

\subsubsection{Pattern models}

We took from \cite{nfm22} some benchmarks based on Dwyer LTL patterns \cite{pattern}.
The considered LTL patterns are the following: \textit{response}, \textit{precedence chain} and \textit{universality} patterns.
The models compose the pattern formulas in two ways: as a sequence of $n$ components
linked in a bus and as a set of components that tries to write on the output port.
These patterns are parametrized on the number of components involved in the composition.

\subsubsection{EVA AUTOSAR contracts}
We took the experimental evaluation from the tool EVA\cite{CimattiTACAS23} and
we adapted it for our compositional reasoning. The models are composed of various
components for which the scheduling is of two types: event-based and cyclic.
Event-based scheduling forces the execution of a component when some of its input
variables change while cyclic scheduling forces the execution of the component every $n$ time units.

In addition to these models,
we also considered a variation of a subset of instances (described in Section~\ref{sec:autosar_ex})
and we relaxed the
scheduling constraints allowing some components to become unresponsive at some point to the original
constraints. We forced an event-based scheduled component (Brake Actuator) to be 
scheduled when its input changes but only until it ends; moreover, we forced
a cyclic scheduled component (Brake Watchdog) to run every $n$ times units until it ends.
Finally, assuming that one of the two components runs infinitely often (i.e. does
not become unresponsive), we check if the composition still holds.

\subsection{Experimental evaluation results}
Figure \ref{fig:comp1} provide an overall comparison between \trr{}, \trrf{} and \trrufa{}
in terms of results and execution time. In these plots, the colour determines the
validity results: blue aggregates all the results without distinction between valid
and invalid; green considers valid instances of both techniques (if the other algorithm returns valid,
timeout is optimistically believed to be true); yellow consider instances in which
the two techniques had different results.

Overall, we checked 624 instances (208 for each rewriting); 406 of these instances
were proven valid, 88 were proven invalid and 60 instances timed out. The general
statistics can be found in Table~\ref{tbl:summary}.
\begin{table}
\begin{filecontents*}{summary.csv}
Algorithm, N. ALL, N. VAL, N. INV, N. UNK, $< 1$ sec, $< 60$ sec, $< 600$ sec
ALL, 624, 476, 88, 60, 82, 448, 535
TrR, 208, 91, 86, 31, 26, 133, 167
TrR+F, 208, 182, 1, 25, 21, 126, 168
TrRUFA, 208, 203, 1, 4, 35, 189, 200
\end{filecontents*}
\csvautobooktabular{summary.csv}
  \caption{Summary result for algorithms}
  \label{tbl:summary}
\end{table}

Table \ref{tbl:res_t} shows some relevant instances of the experimental evaluation
highlighting their execution time and their validity results.

\paragraph{Qualitative results}
As we expected, there are differences in validity results between \trr{} and \trrufa{}.
The pattern models cannot guarantee the validity of the composition with the
weak semantics without additional constraints over the composition.
Intuitively, it suffices that a single component is not scheduled
to violate bounded response properties.

Similarly, we compared the validity results of the 3 rewritings on a simplified version
of the example of Section \ref{sec:simple_ex}; this variation of the model contains
only components $c_1$ and $c_2$.
The simplified version was proved valid using \trrufa{} and \trrf{} while a counterexample
was found using \trr. The simplified version is found invalid when $c_1$ is scheduled
only a finite amount of times failing to deliver the message to component $c_2$.
On the other hand, when we correct the model by adding the component $c_3$, the composition
is proved valid in all the 3 cases.

For what regards the EVA benchmarks, \trr{} and \trrufa{} gave the same validity
results. That occurred because the scheduling constraints were quite strong; 
in particular, the cyclic constraints implicitly forced components, 
such as the watchdog, to run infinitely often. The event-based scheduled component
could have a finite execution but only if their input did not change from some point
on. Differently, updating the constraints by removing these implicit infinite executions
makes the property invalid using \trr{}. To fix these invalid results, we updated the
system assuming that at least one of the two components runs infinitely often and
we relaxed the global property increasing the timed bound for the brake to occur
from 2 time units to 3. Finally, the corrected model was proved in all the 3 cases.

\paragraph{Quantitative results}
We provided a comparison between $\Ropt$ and $\Td$ in term of impact of the
transformation on the verification time.
Since $\Td$ assumes that each local components is executed infinitely often,
we used the rewriting $\Ropt$ equipped with fairness assumption (\trrf{}).
Figure \ref{fig:compfair} shows the comparison between \trrf{} and
\trrufa{}. In this case, \trrufa{} is more
efficient than \trrf. This is not surprising because the generated formula 
assuming infinite local execution can be significantly smaller than the general
one. In general, we see that, regardless of the result, \trrufa{} can prove
the property faster. We think that having to check all the possible combinations
of finite/infinite local execution provides a significant overhead in the verification.

In figure \ref{fig:old_comp}, we compared \trrufa{} with another rewriting for
the composition of temporal properties from \cite{oldrewr}. The comparison has
been done over a subset of the \textit{pattern} models.
The rewriting of \cite{oldrewr} is in principle similar to \trrufa{}; it assumes
infinite executions of local component but, contrary to our approaches, 
in \cite{oldrewr} the asynchronous composition is supported only
considering shared synchronization events\footnote{These events are basically Boolean
variables shared between two components. When one of these variable becomes true, both the components run. Although it is not detailed in the paper, we do support these events
natively by adding additional scheduling constraints in $\alpha$}.
Due to the expressive limitations of the other approach, we applied the evaluation
over a smaller set of instances. It should be noted that the technique of 
\cite{oldrewr} already introduces fairness assumptions of the infinite execution
of local components; therefore, we don't need to manipulate/complicate that rewriting
in this evaluation.
It is clear from the figure that \trrufa{} outperforms this other rewriting.

\section{Conclusions}
\label{sec:conclusions}

In this paper, we considered the problem of compositional reasoning
for asynchronous systems with LTL properties over input and output
variables in which local components are not assumed to run infinitely often.
We introduced a semantics based on the truncated semantics of Eisner and Fisman
for LTL to reason over finite executions of local components.
We proposed a new rewriting of LTL formulas that allows for
checking compositional rules with temporal satisfiability solvers. We then
provided an optimized version of such rewriting.

In the future, we will consider various directions for extending the
framework including real-time and hybrid specifications, optimizations
based on the scheduling of components and other communication
mechanisms such as buffered communication. Moreover, we will generalise our compositional
reasoning for assume/guarantee contracts based on \cite{CT15}, for which we hypothesise
that assumptions must be treated with strong semantics.

\bibliography{biblio} 
\bibliographystyle{alphaurl}
\end{document}